\documentclass[12pt]{article}
\usepackage[dvipsnames]{xcolor}
\usepackage{braket,amsfonts,amsmath,amssymb,amsthm}
\usepackage{mathrsfs,url,hyperref,tikz,graphicx,comment}
\usepackage{pgfplots}
\usepgfplotslibrary{colorbrewer}
\usetikzlibrary{decorations}
\usetikzlibrary{decorations.pathmorphing}
\usetikzlibrary{shapes.geometric}
\pgfplotsset{compat=newest}
\title{How a Space-Time Singularity Helps Remove the Ultraviolet Divergence Problem}

\author{
Joscha Henheik\footnote{Institute for Science and Technology Austria, Am Campus 1, 3400 Klosterneuburg, Austria. E-mail: joscha.henheik@ist.ac.at},~
Bipul Poudyal\footnote{Mathematisches Institut,
	Eberhard-Karls-Universit\"at T\"ubingen, Auf der Morgenstelle 10,
	72076 T\"ubingen, Germany},\footnote{E-mail: 
	bipulpoudyal2015@gmail.com}~~and
Roderich Tumulka$^\dag$\footnote{E-mail:
     roderich.tumulka@uni-tuebingen.de}
}
\date{February 25, 2025}

\newcommand{\I}{\mathrm{i}}
\newcommand{\E}{\mathrm{e}}
\newcommand{\D}{\mathrm{d}}
\newcommand{\Hilbert}{\mathscr{H}}
\newcommand{\Kilbert}{\mathscr{K}}
\newcommand{\sA}{\mathscr{A}}
\newcommand{\sM}{\mathscr{M}}
\newcommand{\Fock}{\mathscr{F}}

\newcommand{\Q}{\mathcal{Q}}
\renewcommand{\Re}{\mathrm{Re}}
\renewcommand{\Im}{\mathrm{Im}}
\newcommand{\PPP}{\mathbb{P}}
\newcommand{\RRR}{\mathbb{R}}
\newcommand{\CCC}{\mathbb{C}}
\newcommand{\SSS}{\mathbb{S}}
\newcommand{\NNN}{\mathbb{N}}
\newcommand{\ZZZ}{\mathbb{Z}}

\newcommand{\ve}{\boldsymbol{e}}

\newcommand{\vv}{\boldsymbol{v}}

\newcommand{\vsigma}{\boldsymbol{\sigma}}
\newcommand{\vomega}{\boldsymbol{\omega}}
\newcommand{\vzero}{\boldsymbol{0}}

\newcommand{\be}{\begin{equation}}
\newcommand{\ee}{\end{equation}}

\newtheorem{prop}{Proposition}
\newtheorem{lemma}{Lemma}
\newtheorem{theorem}{Theorem}

\theoremstyle{definition}
\newtheorem{remark}{Remark}

\begin{document}
\maketitle
\begin{abstract}
	Particle creation terms in quantum Hamiltonians are usually ultraviolet divergent and thus mathematically ill defined. A rather novel way of solving this problem is based on imposing so-called interior-boundary conditions on the wave function. Previous papers showed that this approach works in the non-relativistic regime, but particle creation is mostly relevant in the relativistic case after all. In flat relativistic space-time (that is, neglecting gravity), the approach was previously found to work only for certain somewhat artificial cases. Here, as a way of taking gravity into account, we consider curved space-time, specifically the super-critical Reissner-Nordstr\"om space-time, which features a naked timelike singularity. We find that the interior-boundary approach works fully in this setting; in particular, we prove rigorously the existence of well-defined, self-adjoint Hamiltonians with particle creation at the singularity, based on interior-boundary conditions. We also non-rigorously analyze the asymptotic behavior of the Bohmian trajectories and construct the corresponding Bohm-Bell process of particle creation, motion, and annihilation. The upshot is that in quantum physics, a naked space-time singularity need not lead to a breakdown of physical laws, but on the contrary allows for boundary conditions governing what comes out of the singularity and thereby removing the ultraviolet divergence.
	
	\bigskip
	
	\noindent{\bf Key words:}
	Reissner-Nordstr\"om space-time; naked timelike singularity; interior-boundary condition; particle creation; self-adjoint extension; Bohmian mechanics; regularization of quantum field theory.
\end{abstract}

\section{Introduction}

It is a notoriously difficult problem \cite{Opp30, Nel64, GJ87} (and still active \cite{Pos20, LL23}) to rigorously implement particle creation and annihilation in quantum Hamiltonians at point sources, as they are usually plagued by ultraviolet (UV) divergences. The traditional way of resolving this issue is to employ so-called UV cut-offs (e.g., \cite{Nel64}, see also \cite[Sec.~6.2.5]{Tum22}), corresponding to smearing out the source of particle creation to a positive volume, and (if possible) defining a renormalized Hamiltonian \cite{vH52,Nel64,Der03,GW15} in a limiting procedure removing the cut-off. A different, rather novel approach to this problem is based on \emph{interior-boundary conditions} (IBCs) \cite{TT15a,TT15b}: These relate the wave function $\psi$, defined on a configuration space of a variable number of particles, at the \emph{interior} of the $n$-particle sector to the \emph{boundary} (i.e., where creation/annihilation occurs) of the $n+1$ particle sector. 

The IBC approach has previously successfully been applied in the non-relativistic setting \cite{ibc2a,Lam18}, i.e., for the Schr\"odinger equation involving the Laplacian. However, since particle creation is mostly relevant in the relativistic case, it is of particular importance to study the IBC approach in that setting, for example for the Dirac equation. In flat relativistic space-time (i.e., neglecting gravity), two of us have shown the following \emph{no-go} result (see \cite[Theorem 1]{HT20}): In three spatial dimensions, there exists no self-adjoint Hamiltonian on Fock space that involves particle creation and annihilation at the origin but otherwise acts like the free Dirac Hamiltonian. 
As a positive, but somewhat artificial result \cite[Theorem 6]{HT20}, it was shown that IBC Hamiltonians with particle creation at the origin \emph{can} in fact be implemented in that setting upon adding a sufficiently strong Coulomb potential at the origin. Here, we obtain an IBC Hamiltonian without coupling to a Coulomb potential; we do so by relying only on gravity in a general-relativistic way. In fact, the presence of a space-time singularity makes the IBC approach work without the assumption of a strong Coulomb potential. For further works on IBCs, see \cite{Keppeler_2016,Lienert_2019,IBCdiracCo1,Pos20,BL21}.

In another recent work \cite{HT22}, some of us studied the corresponding Bohmian trajectories and (non-rigorously) constructed a $|\Psi|^2$-distributed Markov jump process (the \emph{Bohm-Bell process} \cite{Bell86,DGTZ04}) in the configuration space of a variable number of particles. Here, we provide the analogous construction with gravity (see Sections~\ref{sec:introBohm} and \ref{sec:2traj}).

In this paper, as a way of taking gravity into account, we consider curved space-time, specifically the super-critical Reissner-Nordström (sRN) space-time \cite{Nor13,Rei16,Wey18,Nor18,HE73}, which is the static curved space-time surrounding a single charged point mass, a solution of the Einstein-Maxwell equations with mass $M\geq 0$, charge $Q\in\RRR$, and angular momentum $0$, where ``super-critical'' means
\be\label{supercritical}
|Q| > M\,.
\ee
More precisely, the super-critical Reissner-Nordstr\"om space-time is given by the manifold $\sM=\RRR\times (\RRR^3\setminus \{\vzero\})\cong \RRR \times (0,\infty) \times \SSS^2$, where $\vzero$ denotes the origin of $\RRR^3$, equipped with the Lorentzian metric $g$ with line element
\be\label{RNmetric}
\D s^2 = A^{2}(r)\, \D t^2 - \dfrac{1}{A^{2}(r)} \,\D r^2 
-r^2 \,\D\vomega^2
\ee
in spherical coordinates ($t \in \RRR ; r \in (0, \infty) ; \vomega \in \SSS^2$). Here, $\D\vomega^2=\D\vartheta^2 + \sin^2 \vartheta\, \D\varphi^2$ in terms of the polar angle $\vartheta$ and the azimuthal angle $\varphi$, and we used natural units $\hbar = c = G= 1$ and the abbreviation 
\begin{equation}\label{Adef}
	A^{2}(r) := 1 - \dfrac{2M}{r} + \dfrac{Q^2}{r^2}
\end{equation}
with parameters $M$ and $Q$ representing the Arnowitt-Deser-Misner (ADM) mass and charge of the metric.
Finally, the electromagnetic four-vector potential is denoted by
\be
A_\mu = (Q/r,0,0,0 )\,,
\ee
not to be confused with the scalar $A$ function introduced in \eqref{Adef}.
In the super-critical regime \eqref{supercritical}, where $A^2(r)>0$ for all $r$, the singularity is timelike and naked (i.e., not surrounded by a horizon), and the metric is static and asymptotically flat. We also take $A(r)>0$. The singularity will be regarded here as the boundary of $\sM$, i.e., $\partial \sM = \{r = 0\} = \RRR \times \{0\} \times \SSS^2$.

We note in passing that the charge and mass values of every charged particle in the standard model of elementary particles and every stable nucleus satisfy the super-criticality condition \eqref{supercritical}, in fact by a large margin of a factor $>10^{15}$, so the classical space-time surrounding an elementary particle would be sRN, provided that the spin does not contribute to the angular momentum of the space-time. While it is not known whether real elementary particles involve space-time singularities, we are in part motivated by the possibility that they might (see Section~\ref{subsec:stsing} for more discussion).

\begin{figure}[h]
\begin{center}
\begin{tikzpicture}[scale=1.8]
\begin{axis}
  [
    axis line style={draw=none},
    tick style={draw=none},
    axis on top=false,
    xtick=\empty,
    ytick=\empty,
    ztick=\empty,
    colormap/Purples-3,
]
\def\d{1};

  \addplot3[surf, shader=flat, draw= black, opacity=0.05, domain=0:14,domain y= 0:40,samples=40,samples y=70]({1.7*x^(2.2)*cos(10*y)},{1.7*x^(2.2)*sin(10*y)},{\d*x});
  
    \addplot3[ surf, shader=flat, draw= black, opacity=0.04, domain=0:4, samples=50]({1.7*x^(2.2)*cos(40*x)},{1.7*x^(2.2)*sin(40*x)},{\d*x});
  
  \addplot3[surf, shader=flat, draw= black, opacity=0.04, domain=4:10,samples=50]({1.7*x^(2.2)*cos(40*x)},{1.7*x^(2.2)*sin(40*x)},{\d*x});
  
    \addplot3[surf, shader=flat, draw= black, opacity=0.04, domain=10:11.9,samples=30]({1.7*x^(2.2)*cos(40*x)},{1.7*x^(2.2)*sin(40*x)},{\d*x});
    
        \addplot3[dotted, surf, shader=flat, draw= black, opacity=0.04, domain=11.9:12.5,samples=15]({1.7*x^(2.2)*cos(40*x)},{1.7*x^(2.2)*sin(40*x)},{\d*x});
 
    \draw[fill=black](0.5,0,15.35) circle (1 pt) node [above] {\tiny $q$,$m$};
  \draw[fill=black](0,0,0) circle (1.2 pt) node [below] {\tiny $Q$,$M$};
  
\end{axis}
\end{tikzpicture}
\vspace{-20mm}
\end{center}
\caption{Qualitative depiction of the setup in this paper: A relativistic quantum mechanical spin-1/2 particle of charge $q$ and mass $m$ moves in a curved space representing the gravitational field of a ``source particle'' with charge $Q$ and mass $M$ (and fixed location, which then is a curvature singularity). The quantum particle can be absorbed and emitted by the source particle. The trajectory shown is a Bohmian trajectory of the quantum particle shortly before absorption or after emission by the source particle.}
	\label{fig:basic}
\end{figure}
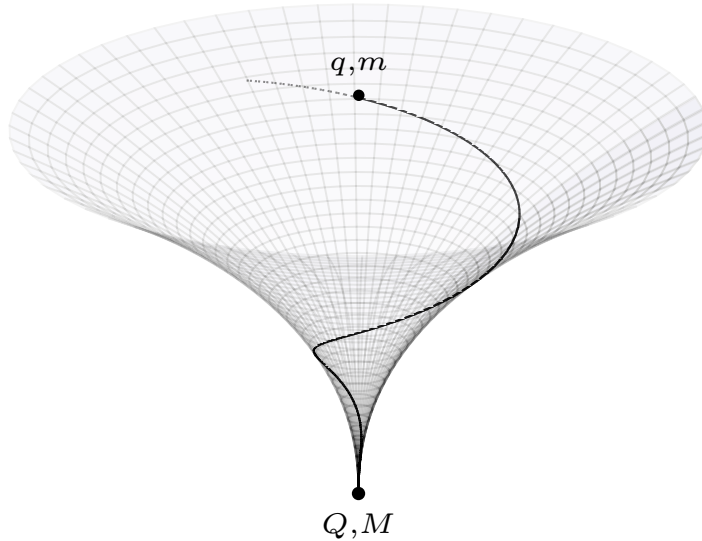

The basic physical picture, illustrated in Figure \ref{fig:basic} and underlying the entire paper, is that a relativistic quantum-mechanical spin-$1/2$ particle of mass $m \ge 0$ and charge $q \in \RRR$ can be emitted and absorbed at the singularity $\partial \sM$. In our setting, this can naturally be associated with a ``source particle" of mass $M \ge 0$ and charge $Q \in \RRR$ obeying \eqref{supercritical} (see Section \ref{subsec:model} for further details). Away from the singularity, the wave function of the quantum particle is governed by the Dirac equation on sRN space-time with Hamiltonian $H_1$ explicitly given in \eqref{DiracHamiltonian} below. (It would be of interest to treat photon wave functions, which have spin 1, but here we focus on spin $1/2$ as the simplest case.)

As our results, briefly described in Section \ref{subsec:resultsbrief} below, we (i) rigorously construct a self-adjoint Hamiltonian $H$ with particle creation, based on IBCs (see Theorem~\ref{theorem3.1} in Section~\ref{sec:thm1}), and (ii) non-rigorously analyze the asymptotic behavior of the Bohmian trajectories close to the space-time singularity in sRN and construct the corresponding Bohm-Bell process (see Section \ref{subsec:jump}, in particular Proposition \ref{prop:2}) for a particular choice of $H$ and ``nice'' wave functions. 

It follows that the quantum particle has nonzero probability to hit the singularity, although the latter could be thought of as a 0d set in 3d space, and the probability of hitting a generic 0d or 1d set vanishes. (The reason for this kind of effective attraction to the singularity $\partial\sM$ is that at $\partial\sM$, the arriving wave function will be transported to the 0-particle sector of Fock space, thereby effectively exerting a kind of suction on the nearby wave function.)

\subsection{Description of Our Main Results} \label{subsec:resultsbrief}

In this section, we briefly describe the main results of the present paper and provide some comments on them. Full details are given in Section \ref{sec:results}. 

\subsubsection{IBC Hamiltonian with Particle Creation} 

As our first main result, Theorem \ref{theorem3.1}, we devise a certain Hamiltonian $H$ with particle creation and annihilation, and prove that it is self-adjoint. As mentioned above, emission/absorption of a particle occurs at a single point in space (or world line in space-time), the naked singularity in sRN space-time \eqref{RNmetric}. Thus, on the one hand, the present work rigorously extends the IBC approach to curved space-time (with fixed background metric), and on the other hand, our treatment deals with (and gives physical meaning to) the well-known fact \cite{CP82} that the 1-particle Dirac Hamiltonian $H_1$ on the sRN space-time is not essentially self-adjoint, and thus does not uniquely define a unitary time evolution. Our Hamiltonian $H$ is based on $H_1$ but is defined on a version of Fock space, as appropriate for particle creation. For simplicity, we consider only the 0-particle and 1-particle sectors of Fock space (but our approach could be extended to the full Fock space along the lines of \cite{ibc2a}). It is common to exclude wave functions of negative energy as unphysical, but we will not exclude them in our model. Our proofs make particular use of mathematical results of Cohen and Powers \cite{CP82} about the domain of the adjoint of $H_1$, in particular described by the asymptotic behavior of wave functions near the singularity. These asymptotics are then exploited to devise an IBC, coupling the 1-particle to the 0-particle sector and thus constituting the Hamiltonian $H$. 

We finally remark that, in \cite[Eq.~(52)]{timelike}, one of us already conjectured an IBC for this case; the IBC investigated here is similar but not identical, and we leave open the question whether a self-adjoint Hamiltonian can be devised for the IBC of \cite{timelike}. For a comparison of the two IBCs, see Remark~\ref{rem:timelike1} in Section~\ref{sec:thm1}.

\subsubsection{Bohmian Trajectories and Bohm-Bell Jump Process}
\label{sec:introBohm}

As our second main result (see Section \ref{subsec:jump}, in particular Proposition \ref{prop:2}), in addition to the Hamiltonian $H$, we construct the Bohm-Bell process \cite{Bell86,DGTZ04} for a particular choice of $H$ (viz., $\widetilde{\kappa}_j=\pm 1$ in the notation of Section~\ref{sec:thm1}) and an initial wave function $\Psi_0$ from a suitable subspace of Hilbert space. It is a piecewise-deterministic Markovian jump process in the configuration space of a variable number of particles that is $|\Psi_t|^2$ distributed at every coordinate time $t$, and its jumps correspond to the creation or annihilation events. Similar processes were devised in \cite{bohmibc} for non-relativistic space-time and in \cite{HT22} for Minkowski space-time with a Coulomb field. While we do not rigorously prove the existence of the process, we can specify what its defining equations must be, in particular the law \eqref{jumprate} for the rate of particle creation at the singularity. This rate depends on the wave function and thus on time, while the direction of emission is uniformly distributed over all directions. A similar law had been conjectured in \cite{timelike}. 

Here is a comparison between the non-relativistic \cite{ibc2a}, the special-relativistic \cite{HT22}, and the present general-relativistic case (summarized in Table~\ref{tab:1}). While the special-relativistic process circles the origin infinitely many times before hitting it, our process does not, and thus is similar in this respect to the non-relativistic process. Another such similarity concerns the radial speed with which the quantum particle hits the origin: while it does with speed 0 in the special-relativistic case, it does with nonzero speed in our and the non-relativistic case. 
(Note that the geometrically appropriate way of measuring this speed is $\D R/\D t$, where $R$ denotes the Regge-Wheeler ``tortoise'' coordinate that makes $(t,R)$ conformally Lorentzian, see Section~\ref{subsec:Rcoord}.)

\begin{table} 
\begin{center}
\begin{tabular}{|r|c|c|c|}
\hline
&non-rel.~\cite{ibc2a}&~~~SR~\cite{HT22}~~~&~GR~(here)~\\\hline\hline
&&&\\
$\displaystyle\frac{\D r}{\D t}(t_0)$ & $\neq 0$ & 0 & $\displaystyle \frac{\D R}{\D t}(t_0)\neq 0$ \\
&&&\\\hline
&&&\\
$\vartheta(t_0)$ & const. & const. & const. \\
&&&\\\hline
&&&\\
$\varphi(t_0)$ & const. & $\to\pm\infty$ & const.\\
&&&\\\hline
\end{tabular}
\end{center}
\caption{Comparison between the Bohm-Bell processes in the non-relativistic, the special-relativistic, and the general-relativistic case; $t_0$ is the time of absorption or emission, and $R$ means the ``tortoise'' (conformally Lorentzian) coordinate defined in \eqref{Rdef}; see \eqref{R(t)} for the relation $R(t)$ and \eqref{eq:asymp} for $\vartheta(t)$ and $\varphi(t)$.}
\label{tab:1}
\end{table}

\subsection{Structure of the Paper} \label{subsec:structure}

The remainder of this paper is organized as follows. In Section~\ref{sec:back}, we put the results into context and provide relevant background information. In Section~\ref{sec:results}, we state our main results. In Section~\ref{chapter2}, we review the known facts about the Dirac equation in the sRN space-time. In Section~\ref{chapter3}, we prove our theorem about the existence of the IBC Hamiltonian. In Section~\ref{chapter4}, we give the details about the construction of the associated Bohm-Bell process. In Section~\ref{sec:conclusions}, we conclude. In Appendix~\ref{app:Phi}, we provide the explicit form of the angular momentum eigenfunctions in a spinor basis corresponding to spherical coordinates.

\section{Motivation, Significance, and Background}
\label{sec:back}

In this section, we further motivate our paper, connect our results to existing literature, and provide additional background information.

\subsection{Space-Time Singularities} \label{subsec:stsing}

One motivation for this research concerns the status of space-time singularities (i.e., of points of infinite space-time curvature): It would seem that the laws of physics break down at singularities, as anything could come out of a singularity if it is timelike (as it is for the sRN space-time). However, in the model considered here, certain laws of nature (the IBC and the law determining the creation rate) govern what comes out of the singularity. That is, the singularity does not lead to a breakdown of physical laws, it provides just the room needed for imposing laws for particle creation and annihilation; this point is discussed further in \cite{timelike}.

Here is how that is related to Roger Penrose's (weak) cosmic censorship conjecture \cite{Penrose1980ge}, which states that naked singularities generically do not form according to general relativity and classical evolution through gravitational collapse from non-singular initial data. Even if that is true, it leaves open whether elementary particles might involve naked singularities, and whether naked singularities might occur in the quantum world. Anyway, we find the possibility of naked singularities worthy of study, in part \emph{because} our results show that they need not entail a breakdown of physical laws, but rather a gap in the physical laws that can be filled by the laws studied here.

\subsection{Ultraviolet Divergence} \label{subsec:UV}

Another aspect concerns the problem of ultraviolet infinities. Terms in a Hamiltonian representing particle creation and annihilation at a point source usually diverge, which keeps the Hamiltonian from being rigorously defined. For example, even in non-relativistic quantum mechanics, the naive Hamiltonian of a quantum particle that can be emitted or absorbed at the origin of 3d Euclidean space is ultraviolet divergent (see, e.g., \cite[Sec.~6.2.5]{Tum22} for discussion). (In particular, the problem arises also if the emitting particle is classical and if emitted particles do not interact.)
Sometimes, renormalization can provide a way of rigorously defining a Hamiltonian \cite{Nel64,Der03} by means of a limiting procedure. Here, we follow a different approach based on IBCs \cite{TT15a,TT15b}, which allow us to directly characterize the Hamiltonian and its domain without a limiting procedure; IBCs are mathematically related to point interactions \cite{AGHKH88,BP35}. We limit ourselves to the (easier) case in which the source (i.e., the emitting and absorbing particle) is fixed at some point (taken to be the coordinate origin). This case was studied for non-relativistic Hamiltonians (based on the Laplacian) in \cite{ibc2a}. 

For the question of whether IBCs can be relevant to realistic quantum field theories, it matters whether they can be applied in a relativistic setting. As a test case, we assume that the particles created are spin-$\tfrac{1}{2}$ particles governed by the Dirac equation. (Photons would be even more interesting, but no general formula is known for their probability current \cite[Sec.~7.3.9]{Tum22}, which is why we prefer the Dirac equation.) It has been shown \cite{HT20} that in Minkowski space-time, IBCs can work in the (somewhat artificial) setting of the Dirac particles feeling a sufficiently strong Coulomb potential around the source, but not if the Coulomb potential is absent or too weak. That sounds not very encouraging; it sounds as if IBCs often failed to work, and as if we should not expect that IBCs could one day be found to work for uncharged relativistic particles such as photons. 

However, the picture changes a lot with the results of the present paper. Basically, the gravitational field of the source (which would also apply to photons) makes the IBC approach work in a similar way as for a strong Coulomb field, regardless of how big the charge $q$ and the mass $m$ of the Dirac particles are. In particular, it also works for uncharged and/or massless particles.\footnote{On the other hand, we use here that the source has sufficiently large charge, $|Q|>M$, but that is, first, actually satisfied for the charge and mass of an electron (as we often think of the sRN space-time as the gravitational field of an elementary particle), and second, it is not so much an issue of the IBC approach as one of the Einstein equation, as the Reissner-Nordstr\"om space-time for $0<|Q|\leq M$ has a complicated structure with infinitely many singularities, wormholes, and asymptotically flat regions \cite[Fig.s 25 and 26(i)]{HE73}, while for $Q=0$ it becomes the Schwarzschild space-time, for which the singularity becomes spacelike and thus not at all like the world line of a particle.} That is, the present paper provides support for expecting the applicability of the IBC approach in more realistic models.

\subsection{Self-adjoint Extensions on Fock Space}
\label{sec:extension}

Mathematically, our problem can be expressed in terms of self-adjoint extensions. This is because, apart from particle creation and annihilation (which happens only at certain places), the Hamiltonian $H$ acts as the Dirac Hamiltonian $H_1$ and we thus devise a self-adjoint extension of $H_1$ to an enlarged Hilbert space, a (truncated) Fock space. (Note, however, that unlike usual self-adjoint extensions, which start from a densely defined operator, $H_1$ in our case is not densely defined, see below.)

In curved space-time, a 1-particle wave function $\psi$ is a cross-section of a smooth complex vector bundle $S$ over $\sM$ (called the spinor bundle) with fibers $S_x$ (called the Dirac spin space) for $x\in\sM$ that are 4-dimensional complex vector spaces. 

For the construction of our $H$,
we are building on previous work on the Dirac Hamiltonian on sRN space-time \cite{CP82,Bel98,BMM00,FSY00,NJ16,KTZT20}, particularly on \cite{CP82}. A crucial difference to these prior works is that, since we consider a mini-Fock space consisting of merely the 0-particle and 1-particle sector, our Hilbert space is 1 dimension larger than what was considered in the prior works: If $\Sigma$ is a $t=\mathrm{const.}$ surface for the Reissner-Nordstr\"om time coordinate $t$, then the prior works considered the 1-particle Hilbert space $\Hilbert^{(1)}$ of functions $\psi:\Sigma\to S$ that are cross-sections (i.e., $\psi(x)\in S_x$) with $\langle \psi,\psi\rangle <\infty$ for the inner product
\be\label{scpdef}
\langle\psi,\phi\rangle = \int_\Sigma V(\D^3x) \, \overline{\psi}(x) \, n_\mu(x) \, \gamma_x^\mu \, \phi(x) \,,
\ee  
where $V$ is the Riemann volume measure arising from the 3-metric on $\Sigma$ and $n_\mu(x)$ the future unit normal vector to $\Sigma$ at $x$ (see \cite[Sec.~7.3.4]{Tum22} for why this is a Hilbert space). Note that the sesquilinear form $(\psi(x),\phi(x))\mapsto \overline{\psi}(x) \phi(x)$ on $S_x$ is (Lorentz invariant and) indefinite of signature $++--$; its coordinate expression is given in \eqref{bardef} below.

In contrast, since we consider particle creation, our Hilbert space is the orthogonal sum
\be\label{Hilbertdef}
\Hilbert=\Hilbert^{(0)}\oplus \Hilbert^{(1)}
\ee
of the 0-particle space $\Hilbert^{(0)}$ and the 1-particle space $\Hilbert^{(1)}$ and thus constitutes a truncated Fock space. The 0-particle space $\Hilbert^{(0)}=\CCC$ is 1-dimensional (because it is spanned by the vacuum state). The Dirac Hamiltonian is at first defined as a differential operator $H_1^0$ on a dense domain $D_1^0$ in $\Hilbert^{(1)}$; while the prior works were studying self-adjoint extensions in $\Hilbert^{(1)}$, we are looking at self-adjoint extensions in $\Hilbert=\Hilbert^{(0)}\oplus \Hilbert^{(1)}$; in particular, the operator $H_1^0$ we extend is densely defined in $\Hilbert^{(1)}$ but not in $\Hilbert$. If $H_1^0$ were \emph{essentially self-adjoint} in $\Hilbert^{(1)}$, it would have a unique self-adjoint extension in $\Hilbert^{(1)}$, and that would be bad for our purpose because it would entail \cite[Theorem 1]{HT20} that all self-adjoint extensions in $\Hilbert$ are block diagonal, which means that no transitions between $\Hilbert^{(0)}$ and $\Hilbert^{(1)}$ ever occur, and thus no particle creation or annihilation takes place. However, as found in \cite{CP82}, $H_1^0$ is not essentially self-adjoint in $\Hilbert^{(1)}$, which gives us room to impose an IBC to obtain a self-adjoint extension $H$ in~$\Hilbert$.

The situation here is different from that in Minkowski space-time: In the latter case, for an uncharged particle ($q=0$) on Euclidean 3-space with one point (say, the origin $\vzero$) removed, the Dirac Hamiltonian is essentially self-adjoint \cite{Sve81}. This roughly means that no probability can flow into or out of the point $\vzero$ and has the consequence \cite{HT20} that no IBC Hamiltonian with particle creation exists. As mentioned in the introduction, that changes when a sufficiently strong Coulomb field is added to the Hamiltonian: then the Dirac Hamiltonian is not essentially self-adjoint, and IBC Hamiltonians exist \cite{HT20}. In the present paper, the action of a Coulomb field on the quantum particle is not necessary (i.e., we can allow $q=0$), as the gravitational field alone already lifts the essential self-adjointness of the Dirac Hamiltonian. In fact, we can even allow $m=0$, and the gravitational field of the
sRN metric with parameters $Q,M$ is still sufficient to ensure that the Dirac Hamiltonian is not essentially self-adjoint, and an IBC Hamiltonian exists.

We do not aim here at identifying all possible IBC Hamiltonians on the sRN space-time; we limit ourselves to a few examples. 

For Reissner-Nordstr\"om space-times in the subcritical regime $|Q|< M$ or the critical regime $|Q|=M$, we expect an IBC to be implementable as well because they have neighborhoods of the singularities that look qualitatively similar to the sRN space-time; however, due to wormholes and several asymptotically flat regions, they are more complicated (and less natural as a model of a point source).

\subsection{Trajectories}
\label{sec:2traj}

We also introduce the natural analog of the Bohm-Bell process for our Hamiltonian $H$ (see Section \ref{subsec:jump}). 
The Bohm-Bell process \cite{Bell86,DGTZ04} is the natural extension of Bohmian mechanics \cite{Bohm52,DT09,Tum22} to include particle creation and annihilation. The process is a Markov process in configuration space; the creation and annihilation events correspond to jumps, as the number of particles changes at these events. Between the jumps, the process is deterministic and follows the Bohmian equation of motion. 

The value of Bohmian mechanics lies in the fact that it provides a realist version of quantum theory \cite{Tum22} while its empirical predictions agree with the standard ones. In fact, Bohmian mechanics resolves the paradoxes and inconsistencies of orthodox quantum mechanics and introduces precision where orthodox quantum mechanics is vague, specifically in the theory of measurement. The Bohm-Bell process that we develop here contributes a further step towards a convincing extension of Bohmian mechanics to quantum field theory.

\subsection{On the Structure of the Model} \label{subsec:model}

Here is how our model fits into a wider class of models. It involves two kinds of particles, let us call them $x$-particles and $y$-particles. The $x$-particles can emit and absorb $y$-particles as in the scheme $x \leftrightarrows x+y$. The $x$-particles have mass $M$ and charge $Q$, the $y$-particles mass $m$ and charge $q$. We treat the $y$-particles quantum-mechanically, whereas the $x$-particles (the sources of emission and absorption) are treated here as non-dynamical and just sit at fixed positions. We limit ourselves to the case of a single $x$-particle and include the classical, general-relativistic gravitational and electromagnetic fields generated by $x$, which is the Reissner-Nordstr\"om space-time \eqref{RNmetric}, considered here for $|Q|>M$ (sRN); recall \eqref{supercritical}. The metric is singular at the location of the $x$-particle, which is why the $x$-particle can be identified with the sRN singularity. If we wanted to treat the $x$-particles quantum-mechanically, while including their general-relativistic gravitational fields, we would presumably need a quantum gravity theory. A non-relativistic IBC-model with quantum-mechanical $x$-particles was defined in \cite{Lam18}.

We remark that our model breaks rotational invariance (which would imply conservation of angular momentum) because under the simplifying assumptions made here, that $x$ has spin 0 and $y$ spin 1/2, local conservation of angular momentum during creation or annihilation events is impossible (already in flat space-time \cite[Sec.~2.4]{HT20}). We expect that IBC Hamiltonians will respect rotational symmetry in more realistic models. The model also violates interaction locality, i.e., the condition that the Hamiltonian contains no interaction between spacelike separated regions. The simple reason is that we allow only 0 or 1 $y$-particles, so once a $y$-particle has been created, and perhaps traveled far away, another creation event at the origin is not possible. We expect that the corresponding model on a full Fock space, allowing all $n\in \NNN\cup\{0\}$ for the number of $y$-particles, will respect interaction locality.

\subsection{Dirac Equation in Curved Space-Time}

There is a standard way of defining the 1-particle Dirac equation in a curved space-time $(\sM,g)$ (see, e.g., \cite{Fock29,IW33,BMB81,PR84,LM89}), which we recall here for convenience of the reader. We also refer to the recent textbook \cite[Chapter 4]{Finsterbook} for an elementary introduction to spinors in curved space-time.

\subsubsection{Coordinate-free Form}

As mentioned already in Section~\ref{sec:extension}, the 1-particle wave function $\psi$ is a cross-section of a vector bundle $S$ over $\sM$ whose fibers $S_x$ are the Dirac spin spaces.
The vector bundle $S$ is equipped with an irreducible representation of the (complexified) Clifford algebra $\mathrm{Cl}_{\mathbb{C}}(T_{x}\sM,g)$ on the spin spaces, $ \mathrm{Cl}_{\CCC}(T_{x}\sM,g) \rightarrow \mathrm{End}(S_{x})$, where $T_x\sM$ is the tangent space at $x\in\sM$; since $T_x\sM$ is itself embedded in the Clifford algebra, the representation includes a linear mapping $\gamma_x: T_x\sM \to \mathrm{End}(S_x)$, called the general-relativistic gamma matrices and subject to the Clifford relation
\be\label{gamma}
\gamma^{\mu}_x\gamma^{\nu}_x + \gamma^{\nu}_x\gamma^{\mu}_x = 2g^{\mu\nu}(x)\,I_x \,,
\ee
where $\gamma^{\mu}_x = \gamma_x(e^{\mu})$ for any basis $e^0,e^1,e^2,e^3$ of $T_x\sM$ and $I_x$ denotes the identity operator in $S_x$. If $(\sM,g)$ is orientable, time-orientable, and possesses spin structure \cite[(1.5.3)]{PR84}, which sRN does \cite[(1.5.6)]{PR84}, then the bundle $S$ and the above-mentioned representation exist; if $\sM$ is simply connected, which the sRN manifold is, then they are unique up to isomorphism \cite[p.~54]{PR84}. We also need the appropriate connection on $S$ or covariant derivative
\be
\nabla : \Gamma(S) \rightarrow \Gamma(T^{*}\sM \otimes S)\,,
\ee
where $\Gamma(S)$ denotes the set of smooth cross-sections of the bundle $S$; $\nabla$ is uniquely defined by the metric \cite[Sec.~4.4]{PR84}.
The 1-particle Dirac equation in $(\sM,g)$ is then
\be\label{DEa}
\bigl( \I\gamma^{\mu}_x\nabla_{\mu}- q \gamma^\mu_x A_\mu(x) - m \bigr)\psi(x) = 0\,,
\ee
where $m\geq 0$ is called the mass of the particle and $q\in \RRR$ its charge. Finally, $S_x$ is equipped with an ``overbar'' operation $\psi\mapsto \overline{\psi}$, a conjugate-linear mapping from $S_x$ to its dual space $S_x^*$, and the Born distribution (``$|\psi|^2$ distribution'') on a spacelike surface $\Sigma$ is the measure given by
\be\label{Born1}
n_\mu(x) \, j^\mu(x) \, V(\D^3 x)
\ee
with the probability current 4-vector field
\be\label{jdef}
j^\mu(x) = \overline{\psi}(x) \, \gamma_x^\mu \, \psi(x) \,.
\ee

\subsubsection{Expression in Spherical Coordinates} \label{subsubsec:sphcoord}

The Dirac equation in the sRN space-time has been studied before in many works, e.g., \cite{CP82,Bel98,FSY00,BMM00,NJ16,KTZT20}. We adopt a widely used basis $b_x$ in spin space $S_x$ defined as follows: From the coordinate basis\footnote{We follow here the convention of identifying a tangent vector with the directional derivative operator in that direction.} $(\partial_t,\partial_r, \partial_\vartheta,\partial_\varphi)$ of $T_x\sM$, we obtain an orthonormal basis (Lorentz frame) $e_x$ by normalizing the vectors, 
\be\label{edef}
e_x=(e_x^0,e_x^1,e_x^2,e_x^3)=\bigl(A^{-1}\partial_t, A\partial_r,r^{-1}\partial_\vartheta, (r\sin \vartheta)^{-1}\partial_\varphi\bigr)\,.
\ee
To this orthonormal basis there corresponds a basis $b_x$ of $S_x$; the correspondence is canonical up to an overall sign which we choose continuously in $x$; $b_x$ is an orthonormal basis relative to the scalar product $\overline{\psi} \, \gamma^\mu(x) \, g_{\mu\nu}(x)\, n^\nu(x) \, \phi$ in $S_x$ associated with the surface $\{t=\mathrm{const.}\}$ or its future unit normal vector $n(x) = A^{-1} \partial_t$. Relative to the bases $e_x$ and $b_x$, the gamma matrices have their standard entries \cite{Tha92},
\be
\gamma^0 = \begin{pmatrix} I_2 & 0\\ 0 & -I_2 \end{pmatrix}\,,~~
\gamma^i = \begin{pmatrix} 0 & \sigma_i\\ -\sigma_i & 0 \end{pmatrix}
\ee
with $\sigma_i$ the $i$-th Pauli matrix. That is, in these bases the general-relativistic gamma matrices $\gamma_x^\mu$ reduce to the \emph{special-relativistic} gamma matrices, to which the symbol $\gamma^\mu$ will henceforth refer. Likewise, in the basis $b_x$, the overbar operation is represented in the same way as in any Lorentz frame, 
\be\label{bardef}
\overline{\psi}=\psi^\dagger \gamma^0\,.
\ee 

The Hilbert space $\Hilbert^{(1)}$ of 1-particle wave functions on $\Sigma=\{t=0\}$ can therefore be represented in coordinates as
\be\label{Hilbert1def}
\Hilbert^{(1)} = L^2\Bigl((0,\infty)\times \SSS^2,\CCC^{4}, A^{-1}r^2 \, \D r\, \D^{2}\vomega\Bigr)
\ee
with $\D^2\vomega = \sin\vartheta \, \D\vartheta \, \D\varphi$ and inner product
\be\label{scpcoo}
\langle \psi,\phi\rangle = \int_0^\infty \D r \int_{\SSS^2} \D^2 \vomega \, A^{-1}\, r^2 \, \psi(r,\vomega)^\dagger \, \phi(r,\vomega)\,.
\ee
(Note for comparison that $L^2$ of 3d Euclidean space is equivalent to $L^2((0,\infty)\times \SSS^2,r^2 \, \D r \, \D^2\vomega)$ in spherical coordinates.) Indeed, \eqref{scpcoo} follows from the general expression \eqref{scpdef}, as the Riemannian volume measure $V$ on $\Sigma$ has in general density $\left|\det {}^3g\right|^{1/2}$ in coordinates and is in this case given by $V(\D r \times \D^2 \vomega) =  A^{-1}r^2\,\D r\, \D^{2}\vomega$, while $n_\mu= (1,0,0,0)$ in the basis $e_x$. 

Correspondingly, the Born distribution is given in coordinates by
\be\label{Born2}
|\psi(r,\vomega)|^2 \, A^{-1}\, r^2\, \D r \, \D^2 \vomega\,,
\ee
where $|\psi|^2$ means $\psi^\dagger \psi$ or, equivalently, the sum of the absolute squares of the four complex components of $\psi$.

The Dirac equation on sRN space-time then reads in coordinates
\be
\I \partial_t \psi = H_1 \psi
\ee
with Hamiltonian \cite[Eq.\ between (2.7) and (2.8)]{CP82}\footnote{Cohen and Powers \cite{CP82} by mistake wrote $\tfrac{1}{2}\cos \theta$ for $\tfrac{1}{2}\cot \theta$ in that equation; when comparing, note also that they used the notation $\gamma^\mu$ for our $-\I\gamma^\mu$ and $q$ for our $-qQ$.}
\begin{align}
H_1 
&= -\I\alpha^1 A^2(\partial_r + r^{-1}+\tfrac{1}{2}A^{-1}A')
-\I\alpha^2 r^{-1}A(\partial_\vartheta+\tfrac{1}{2}\cot\vartheta)\nonumber\\
&\quad -\I\alpha^3 (r\sin\vartheta)^{-1} A\partial_\varphi
+ mA \beta+qQr^{-1},\label{DiracHamiltonian}
\end{align}
where $A'=\partial_r A$ is the derivative of $A$ and, as usual, $\beta=\gamma^0$ and $\alpha^i=\gamma^0 \gamma^i$. This operator is defined on $C_c^\infty((0,\infty)\times \SSS^2,\CCC^4)$, the space of smooth functions with compact support, which is a dense subspace of $\Hilbert^{(1)}$, and $H_1$ is, in fact, symmetric relative to \eqref{scpcoo} on this subspace, as will follow from Lemma~\ref{lem:unitarytransform} below or can be checked through direct computation.

\section{Main Results}
\label{sec:results}

In this section, we formulate our main results. 
Recall that the Hilbert space of our model is the mini-Fock space $\Hilbert=\Hilbert^{(0)}\oplus \Hilbert^{(1)}$ corresponding to 0 or 1 $y$-particles (see Section \ref{subsec:model} for the terminology of $x$- and $y$-particles) with $\Hilbert^{(0)}=\CCC$ and $\Hilbert^{(1)}$ given by \eqref{Hilbert1def}. 
The corresponding configuration space is
\be\label{confdef}
\Q= \Q^{(0)} \cup \Q^{(1)} = \{\emptyset\} \cup \Sigma\,. 
\ee
Here, $\emptyset$ is the 0-particle configuration and $\Sigma$ is any one of the $\{t=\mathrm{const.}\}$ surfaces; these surfaces can be identified with each other in a canonical way (as the mapping connecting points with equal $(r,\vartheta,\varphi)$ coordinates is an isometry) and represented in coordinates $(r,\vomega)$ as the Riemannian 3-manifold
\be
\Sigma=(0,\infty)\times \SSS^2
\ee
with the metric
\be
\D s^2 = A^{-2} \D r^2 + r^2 \D\vomega^2\,.
\ee
The Born distribution on $\Q$ for $\Psi\in\Hilbert$ is the measure $\PPP_\Psi$ with
\begin{subequations}\label{Born3}
\begin{align}
\PPP_\Psi(\{\emptyset\})
&= |\Psi^{(0)}|^2 \label{Born3a} \\
\PPP_\Psi(\D r \times \D^2 \vomega) 
&= |\Psi^{(1)}(r,\vomega)|^2 \, A^{-1} \, r^2 \, \D r\, \D^2\vomega\,. \label{Born3b}
\end{align}
\end{subequations}

\subsection{IBC Hamiltonian with Particle Creation}
\label{sec:thm1}

In order to formulate our first main result, the existence of the Hamiltonian, we use a certain orthonormal basis of $L^2(\SSS^2,\CCC^4,\D^2\vomega)$ traditionally denoted $\Phi^{\pm}_{m_j,\kappa_j}$, where $(m_j,\kappa_j)$ varies in the set
\be\label{sAdef}
\sA:= \Bigl\{(m_j,\kappa_j):\kappa_j \in \ZZZ\setminus\{0\},~ m_j+\tfrac{1}{2}\in\ZZZ,~ |m_j| \leq |\kappa_j|-\tfrac{1}{2} \Bigr\}\,.
\ee
Without going into details (see Appendix~\ref{app:Phi} or \cite[Sec.~4.6.4]{Tha92} for the definition), we remark that the $\Phi^{\pm}_{m_j\kappa_j}$ are the joint eigenfunctions of $\pmb{J}^{2}, J_z, K$, and $\beta$, viz.,
\begin{subequations}\label{Phieigenvectors}
\begin{align}
    \pmb{J}^{2}\Phi^{\pm}_{m_{j}\kappa_{j}} 
    &= j(j+1)\Phi^{\pm}_{m_{j}\kappa_{j}}\\
    J_{3}\Phi^{\pm}_{m_{j}\kappa_{j}} 
    &= m_{j}\Phi^{\pm}_{m_{j}\kappa_{j}}\\
    K\Phi^{\pm}_{m_{j}\kappa_{j}} 
    &= \kappa_{j}\Phi^{\pm}_{m_{j}\kappa_{j}}\\
    \beta \Phi^{\pm}_{m_j\kappa_j} 
    &= \pm \Phi^{\pm}_{m_j\kappa_j}
\end{align}
\end{subequations}
with $j=|\kappa_j|-\tfrac{1}{2}$, where (again without going into details) $\pmb{J}=\pmb{L}+\pmb{S}$ is the triple of angular momentum operators, $\pmb{L}$ the orbital angular momentum, $\pmb{S}$ the spin angular momentum, and $K = \beta (2 \pmb{S}\cdot\pmb{L}+1)$ the spin-orbit operator.

We also note for use in the IBC \eqref{ibc} that since the $\beta$ matrix has eigenvalues $\pm 1$, the projection to the eigenspace with eigenvalue $-1$ is $\tfrac{1}{2}(I-\beta)$.

We define a Hamiltonian $H$ in $\Hilbert$ for every choice of $(\widetilde{m}_j,\widetilde{\kappa}_j)\in\sA$ and of a coupling constant $g\in\CCC\setminus\{0\}$; $H$ acts 
on wave functions subject to the interior-boundary condition
\be\label{ibc}
\lim_{r\searrow 0} \tfrac{1}{2}(I-\beta)r^{1/2} \, \Psi^{(1)}(r,\vomega) = g\, |Q|^{-1/2} \Phi^-_{\widetilde{m}_j\widetilde{\kappa}_j}(\vomega) \, \Psi^{(0)}~~~~\forall\vomega\in\SSS^2
\ee
according to 
\begin{subequations}
\begin{align}
(H\Psi)^{(0)}&= g^* \, |Q|^{1/2}\lim_{r\searrow 0} \int_{\SSS^2}\!\! \D^2 \vomega \, \Phi^+_{\widetilde{m}_j\widetilde{\kappa}_j}(\vomega)^\dagger \, r^{1/2} \Psi^{(1)}(r,\vomega) \label{Hact0}\\
(H\Psi)^{(1)}(r,\vomega)&= H_1\Psi^{(1)}(r,\vomega)~~~(r>0)\label{Hact1}
\end{align}
\end{subequations}
with $H_1$ the Dirac Hamiltonian as in \eqref{DiracHamiltonian} and $g^*$ the conjugate of the complex constant $g$.

Here is the precise statement about the IBC Hamiltonian $H$:

\begin{theorem}[IBC Hamiltonian with particle creation]\label{theorem3.1}
For every choice of the parameters $(\widetilde{m}_j,\widetilde{\kappa}_j)\in\sA$ and  $g\in\CCC\setminus\{0\}$, there is a self-adjoint operator $H$ with domain $D\subset\Hilbert$ such that 
\begin{enumerate}
\item For every $\Psi\in D$, the upper sector is of the form
\be \label{eq:asexp}
\Psi^{(1)}(r,\vomega) = f(\vomega) \, r^{-1/2} + \mathcal{O}(r^{1/2})
\ee
as $r\to 0$ for some (uniquely determined, $\Psi$-dependent) $f\in L^2(\SSS^2,\CCC^4,\D^2\vomega)$. In particular, the limit on the left-hand side of \eqref{ibc} exists and is the part of $f$ in the eigenspace of $\beta$ with eigenvalue $-1$.
\item Every $\Psi\in D$ satisfies the IBC \eqref{ibc}.
\item For every $\Psi^{(1)}\in C_c^\infty((0,\infty)\times \SSS^2, \CCC^4)$, $(0,\Psi^{(1)})\in D$, and $H(0,\Psi^{(1)})=(0,H_1 \Psi^{(1)})$ with $H_1$ as in \eqref{DiracHamiltonian}. Put differently, $(H,D)$ is a self-adjoint extension of $(H_1^0,D^0)$ with $D^0= \{0\}\oplus C_c^\infty((0,\infty)\times \SSS^2, \CCC^4)$ and $H_1^0(0, \psi)=(0,H_1\psi)$.
\item The $0$-particle action of $H$ is given by \eqref{Hact0}, which holds in the precise sense that
\be
(H\Psi)^{(0)}= g^* |Q|^{1/2} \langle \Phi^+_{\widetilde{m}_j\widetilde{\kappa}_j},f\rangle_{L^2(\SSS^2,\CCC^4,\D^2\vomega)} \,. 
\ee
\item Particle creation occurs, i.e., $H$ is not block diagonal in the decomposition $\Hilbert^{(0)}\oplus \Hilbert^{(1)}$.
\end{enumerate}
\end{theorem}

Theorem \ref{theorem3.1} will follow as a special case of the slightly reformulated and more general Theorem \ref{thm:thmref}, formulated in Section \ref{subsec:thmref}. We give the proof of Theorem \ref{thm:thmref} in Section~\ref{sec:pfthm1}.

\begin{remark}[Boundary conditions for the Dirac equation]
While a boundary condition for the Laplacian usually specifies the value of $\psi$ on the boundary (as in a Dirichlet boundary condition) or its normal derivative (as in a Neumann boundary condition), boundary conditions for the Dirac equation usually specify two of the four components of the wave function on the boundary, leaving the other two unspecified (e.g., \cite{FR15}). Likewise, except for the scaling factor $r$ (which has to do with how to extend the bundle $S$ to the boundary \cite[Sec.~5.3]{timelike}), \eqref{ibc} specifies two of the four components of $\Psi^{(1)}$ at $r=0$ (those in the eigenspace of $\beta$ with eigenvalue $-1$), leaving the other two unspecified (those in the eigenspace of $\beta$ with eigenvalue $+1$).
\end{remark}

\begin{remark}[Comparison to \cite{timelike}]\label{rem:timelike1} 
In \cite{timelike}, one of us conjectured what a Hamiltonian on a sRN space-time with an IBC at the singularity and the corresponding Bohm-Bell process might look like. The description there was based on plausibility rather than rigorous analysis, but gets qualitatively confirmed by Theorem~\ref{theorem3.1} above. Since our proof technique for Theorem~\ref{theorem3.1} makes use of the angular momentum eigenspaces $\Kilbert_{m_j\kappa_j}$ spanned by $\Phi^+_{m_j\kappa_j}$ and $\Phi^-_{m_j\kappa_j}$, while the IBC and Hamiltonian in \cite{timelike} were not related to these subspaces, the $H$ provided by Theorem~\ref{theorem3.1} is not the same as the one described in \cite{timelike}, and we cannot answer whether the equations in \cite{timelike} for the IBC and the action of the Hamiltonian do or do not define a self-adjoint operator. But the $H$ of Theorem~\ref{theorem3.1} is similar to the one described in \cite{timelike} in that (i)~the IBC \eqref{ibc}, just as \cite[(52)]{timelike}, concerns two components of the limiting values of $\Psi^{(1)}$ on the singularity, rescaled by $r^{1/2}$, and requires them to be $\Psi^{(0)}$ times a certain spinor function of $\vomega$; (ii)~the expression \eqref{Hact0} for $(H\Psi)^{(0)}$, just as \cite[last line of (53)]{timelike}, is the inner product over $\SSS^2$ of the rescaled $\Psi^{(1)}$ at the singularity with another spinor function of $\vomega$; and (iii)~$H$ acts like the Dirac Hamiltonian away from the singularity.
\end{remark}

\begin{remark}[Comparison to \cite{HT20}]\label{rem:HT20}
In \cite[Thm.~6]{HT20}, two of us proved the existence of a self-adjoint IBC Hamiltonian in flat Minkowski space-time under the assumption of a sufficiently strong Coulomb potential acting on the quantum particle. Some elements of the construction and the proofs were similar; some differences are that the asymptote \eqref{eq:asexp} of $\Psi^{(1)}$ as $r\to 0$ had a different form involving different powers of $r$, thus requiring a different power of $r$ in the IBC; that only few choices of $\widetilde{m}_j,\widetilde{\kappa}_j$ worked; and the IBC involved a different spinor field instead of $\Phi^-_{\widetilde{m}_j\widetilde{\kappa}_j}$.
\end{remark}

\begin{remark}[Full Fock space]\label{rem:Fock}
Along the lines of \cite{ibc2a}, our construction could be extended to full Fock space $\Fock$ with an arbitrary number $n\in\{0,1,2,3,\ldots\}$ of $y$-particles. For each value $\tau\in\RRR$ of the time coordinate $t$, let $\Sigma_\tau:=\{t=\tau\}$ and the configuration space be $\Q_\tau:=\bigcup_{n=0}^\infty \Sigma_\tau^n$. The boundary of configuration space consists of those configurations with at least one $y$-particle at $r=0$, and the IBC will relate the $n$-particle sector $\psi^{(n)}$ of $\psi\in\Fock$ to the values of $\psi^{(n+1)}$ on the boundary.
\end{remark}

\begin{remark}[Multi-time wave functions]\label{rem:multitime}
In relativistic space-time, it is usually possible and of interest to extend the domain of definition of wave functions so as to make them \emph{multi-time} wave functions \cite{LPT20,Lienert_2019}, i.e., defined not only for simultaneous $n$-particle configurations but for any spacelike $n$-particle configuration or even any $n$-particle configuration at all. This is also possible for the present model, including states of arbitrary particle number $n$ as in Remark~\ref{rem:Fock}, but the $x$-particle, serving as the source at the singularity, needs to be taken into account: while it cannot occupy other locations than the origin $r=0$, it should be given its own time variable $t_x$ in a multi-time approach, leading to wave functions of the form
\be\label{multitime}
\psi^{(n)}(t_x,t_1,r_1,\vartheta_1,\varphi_1,\ldots, t_n,r_n,\vartheta_n,\varphi_n)\,,
\ee
where $t_j,r_j,\vartheta_j,\varphi_j$ are the coordinates of the $j$-th $y$-particle. Since for multi-time wave functions, the space-time points of two interacting particles need to be spacelike separated, each $(t_j,r_j,\vartheta_j,\varphi_j)$ is constrained to the region spacelike from $(t_x,r=0)$ (shaded in Figure~\ref{fig:spacelike}). In fact, for any $n$ points in this region, the function \eqref{multitime} is uniquely determined from the $n$-particle wave function on $\Sigma_{t_x}$ (provided by the single-time evolution) as the solution of the free Dirac equation in each $(t_j,r_j,\vartheta_j,\varphi_j)$ away from the singularity.
\end{remark}

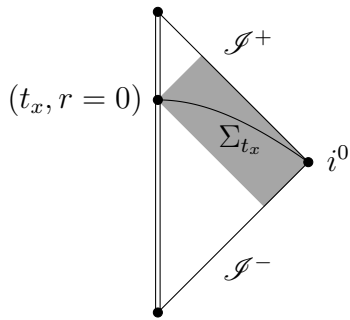
\begin{figure}[h]
\begin{center} 
  \begin{tikzpicture}
  \draw[fill=gray!70!white,draw=gray!70!white] (0,0.8284) -- (1.4142,-0.5858) -- (2,0) -- (0.5858,1.4142) -- cycle;
  \draw [domain=0:2, samples=50] plot (\x,{-2*(sqrt(1+\x*\x/4)-sqrt(2))});
  \draw (0.03,2) -- (0.03,-2);
  \draw (-0.03,2) -- (-0.03,-2);
  \draw (0,2) -- (2,0) -- (0,-2);
  \filldraw (0,0.8284) circle [radius=0.06];
  \filldraw (0,-2) circle [radius=0.06];
  \filldraw (0,2) circle [radius=0.06];
  \filldraw (2,0) circle [radius=0.06];
  \node at (-1.1,0.8284) {$(t_x,r=0)$};
  \node at (2.4,0.05) {$i^0$};
  \node at (1.2,1.6) {$\mathscr{I}^+$};
  \node at (1.2,-1.4) {$\mathscr{I}^-$};
  \node at (1.1,0.3) {$\Sigma_{t_x}$};
  \end{tikzpicture}
\end{center}
 \caption{Penrose conformal diagram of sRN space-time $\sM$, shown with the spacelike coordinate surface $\Sigma_{t_x}=\{t=t_x\}$ bordering on the point $(t_x,r=0)$ on the singularity $\partial \sM=\{r=0\}$ (shown as the vertical double line); the value of $t_x$ was chosen arbitrarily; $\mathscr{I}^{\pm}$ is the future (past) null infinity, $i^0$ is the spacelike infinity, and the shaded region comprises the points spacelike separated from $(t_x,r=0)$.}\label{fig:spacelike}
\end{figure}

\subsection{Bohmian Trajectories and Bohm-Bell Jump Process} \label{subsec:jump}

In Theorem \ref{theorem3.1}, we constructed a self-adjoint Hamiltonian involving the creation of Dirac particles at the sRN singularity using an IBC \eqref{ibc}. In this section, we construct a Markov process $Q_t$ (the ``Bohm-Bell process'') in the configuration space $\Q$ as in \eqref{confdef} that is Born (``$|\Psi_t|^2$'') distributed at every $t\in\RRR$. Our approach is analogous to ``Bell-type quantum field theory'' \cite{Bell86} in which motion of the configuration along deterministic trajectories are interrupted by stochastic jumps. That is, $Q_t$ follows Bohmian trajectories between the jumps, and the latter correspond to the creation/annihilation of particles. A key element of this construction is to determine the rate of particle creation that ensures equivariance of the process (i.e., preservation of the Born distribution), and to this end it is relevant to determine the asymptotic Bohmian trajectories near the singularity for this Hamiltonian.

\subsubsection{Bohmian Trajectories} 
\label{subsubsec:trajectories}

We now review the definition of the Bohmian trajectories and determine the coordinate form of their equation for our setup.

The Bohmian trajectories $X:\RRR\to\sM$ are solutions to Bohm's equation of motion \cite{Bohm52} for the Dirac equation \cite{Bohm53}, given by 
\begin{equation} \label{eq:BohmEM}
\frac{\D X^\mu}{\D s}\propto j^\mu(X(s))\,,
\end{equation}
where $s$ is any curve parameter and $j^\mu$ given by \eqref{jdef}. In words, the world line is everywhere tangent to the vector field $j^\mu$.

We now want to express the equation of motion in coordinates. A subtle point is that there are two relevant bases in the tangent space $T_x\sM$ in which $j(x)$ can be expressed: the coordinate basis $(\partial_t,\partial_r,\partial_\vartheta,\partial_\varphi)$ and the basis $e_x$ of \eqref{edef} (the normalized coordinate basis). We write $(j^t,j^r,j^\vartheta,j^\varphi)$ for the components of $j(x)$ relative to the former and $(j^0,j^1,j^2,j^3)$ for those relative to the latter,
\begin{subequations}
\begin{align}
j(x) &= j^t\partial_t +j^r\partial_r +j^\vartheta\partial_\vartheta +j^\varphi\partial_\varphi\\
j(x) &= j^0e^0_x + j^1 e^1_x + j^2e^2_x + j^3 e^3_x\,.
\end{align}
\end{subequations}
One can read off from \eqref{edef} that $j^t=j^0A^{-1}$, $j^r=j^1A$, $j^\vartheta=j^2r^{-1}$, and $j^\varphi=j^3(r\sin \vartheta)^{-1}$. Since the world line $X$ is tangent to the vector field $j$ on $\sM$, the image of the world line in coordinate space with axes $t,r,\vartheta,\varphi$ is tangent to the image of $j$, which has components $(j^t,j^r,j^\vartheta,j^\varphi)$. Therefore, \eqref{eq:BohmEM} reduces to
\begin{equation} \label{eq:Bohmcoo1}
\frac{\D}{\D t} \begin{pmatrix} r(t)\\ \vartheta(t) \\ \varphi(t) \end{pmatrix} 
=\vv\bigl(t,r(t),\vartheta(t),\varphi(t)\bigr)
\end{equation}
with
\begin{subequations}\label{Bohmcoo2}
\begin{align}
v^1&= \frac{j^r}{j^t}= \frac{j^1A}{j^0A^{-1}} = A^2 \frac{(\Psi^{(1)})^\dagger \alpha^1 \Psi^{(1)}}{|\Psi^{(1)}|^2}\\
v^2&= \frac{j^\vartheta}{j^t}= \frac{j^2r^{-1}}{j^0A^{-1}} = \frac{A}{r}\, \frac{(\Psi^{(1)})^\dagger \alpha^2 \Psi^{(1)}}{|\Psi^{(1)}|^2}\\
v^3&= \frac{j^\varphi}{j^t}= \frac{j^3(r\sin\vartheta)^{-1}}{j^0A^{-1}} = \frac{A}{r\sin\vartheta}\, \frac{(\Psi^{(1)})^\dagger \alpha^3 \Psi^{(1)}}{|\Psi^{(1)}|^2}\,.
\end{align}
\end{subequations}

\subsubsection{Asymptotics of Bohmian Trajectories} 

We now determine the asymptotic form of the trajectories just before reaching (or after emanating from) the singularity.
In the following, we assume that
\be
\widetilde{\kappa}_j=\pm 1\,,
\ee
and we will only consider wave functions $\Psi$ lying in a certain subspace $\widehat{D}\subset \Hilbert$ which is invariant under the time evolution generated by $H$. More precisely, we take $\widehat{D}$ to be the part of the domain $D$ of $H$ whose $1$-particle component $\Psi^{(1)}$  has angular momentum corresponding to the chosen $(\widetilde{m}_j, \widetilde{\kappa}_j) \in \sA$. In the notation of Sections~\ref{subsec:Rcoord} and \ref{proof1},
\be\label{hatDdef}
\widehat{D}:=(1\oplus U^{-1}) \widehat{D}_{\widetilde{m}_j\widetilde{\kappa}_j}
\ee
involving the unitary transformation $U$ as in \eqref{unitarytransform} and the subspace $\widehat{D}_{\widetilde{m}_j\widetilde{\kappa}_j}$ as in \eqref{eq:hatD}.
As becomes apparent from the proof of Theorem \ref{theorem3.1} in Section \ref{chapter3}, the coupling between $\Hilbert^{(0)}$ and $\Hilbert^{(1)}$ happens only within $\widehat{D}$, hence $\widehat{D}$ is the most relevant or interesting part of $D$. Thus, by focusing on $\widehat{D}$, we avoid unnecessarily tedious computation for extracting the qualitative behavior, which we believe will not change much for general $\Psi\in D \setminus \widehat{D}$; cf.~\cite{HT22}.

Moreover, in the following asymptotic analysis of the Bohmian trajectories, we will also make use of a (plausible and common \cite{bohmibc, HT22}) approximation for Bohm's equation of motion: We assume that the Bohmian velocity field $\vv$ as in \eqref{Bohmcoo2} varies slowly in time. More specifically, we assume that for times $t$ close to the reference time $t_0 \in \RRR$, the asymptotics of the true Bohmian trajectories as solutions of \eqref{eq:Bohmcoo1} coincide (to leading order) with those one would obtain from a time-independent velocity field $\vv(t_0,\cdot)$, i.e., with solutions of 
\begin{equation} \label{eq:BohmEMapp}
\frac{\D}{\D t} \begin{pmatrix} r(t)\\ \vartheta(t) \\ \varphi(t) \end{pmatrix} 
=\vv\bigl(t_0,r(t),\vartheta(t),\varphi(t)\bigr) \,.
\end{equation}
This approximation corresponds to approximating $\Psi_t \approx \Psi_{t_0}$ in a suitable topology; see \cite[Remark~1]{HT22} for a possible general strategy of rigorously justifying it. 

This is our main result on the asymptotics of the Bohmian trajectories.

\begin{prop}[Asymptotics of Bohmian trajectories]  \label{prop:2}
	Let $\widetilde{\kappa}_j=\pm 1$, $\Psi_0 \in \widehat{D}$, denote the time-evolved state by $\Psi_t := \E^{-\I Ht} \Psi_0 \in \widehat{D}$ and write 
	\be\label{cdef}
	c_{\pm}(t)=|Q|^{1/2}\langle \Phi^{\pm}_{\widetilde{m}_j\widetilde{\kappa}_j},f_t \rangle= |Q|^{1/2} \lim_{r\searrow 0}\int_{\SSS^2}\D^2\vomega \, \Phi^{\pm}_{\widetilde{m}_j\widetilde{\kappa}_j}(\vomega)^\dagger \, r^{1/2}\, \Psi^{(1)}_t(r,\vomega)\,,
	\ee
	where $f_t$ is the analog of $f$ from \eqref{eq:asexp} obtained from $\Psi_t$. 
	Let $t_0 \in \RRR$ be any time for which
	\be\label{assumption_c}
	\Im[c_-^*(t_0) \, c_+(t_0) ] \neq 0 
	\ee
	and abbreviate $c_{\pm}:= c_{\pm}(t_0)$. 
	
	Then the solution to \eqref{eq:BohmEMapp} with $r(t_0) = 0$, i.e., the trajectories emanating from/reaching the singularity at time $t_0$, occur only if $\Im[c_-^* c_+] < 0$ (resp. $\Im[c_-^* c_+] >0$) and they obey for $t>t_0$ (resp. $t<t_0$) the following asymptotics as $t\to t_0$:
	\begin{subequations} \label{eq:asymp}
		\begin{align}
			&r(t) = C_\mathrm{rad} \, |t-t_0|^{1/3} + \mathcal{O}\Bigl(|t-t_0|^{2/3} \Bigr) \label{eq:r(t)}
			\\[2mm]
			&\vartheta(t) =  \vartheta_{0} + \mathcal{O}\Bigl( |t-t_0|^{2/3} \Bigr) \\[2mm] 
			&\varphi(t) =  \varphi_0 + \mathrm{sgn}(t-t_0) \, C_\mathrm{az} \, |t-t_0|^{1/3} +  \mathcal{O}\Bigl( |t-t_0|^{2/3} \Bigr) \label{eq:phi(t)}
		\end{align}
	\end{subequations}
	for some constants $\vartheta_0 \in [0,\pi]$ and $\varphi_0 \in [0,2\pi)$ and with coefficients
	\begin{subequations}\label{Cdef}
	\begin{align}
	C_\mathrm{rad}&= \Biggl(\frac{6Q^2 \,  |\Im [c_-^*c_+]|}{|c_+|^2 + |c_-|^2}\Biggr)^{\! 1/3}\\
	C_\mathrm{az}&= \frac{6^{1/3}\mathrm{sgn}(Q\widetilde{m}_j \widetilde{\kappa}_j) \, \Re[c_-^*c_+]}{|Q|^{1/3}(|c_+|^2+|c_-|^2)^{1/3}  |\Im [c_-^*c_+]|^{2/3}}\,.
	\end{align}
	\end{subequations}
	Moreover,
	\be\label{phi(r)}
	\varphi(r)=\varphi_0-\frac{\mathrm{sgn}(\widetilde{m}_j \widetilde{\kappa}_j)}{Q} \frac{\Re[c_-^{*}c_+]}{\Im [c_-^*c_+]} \: r+ \mathcal{O}(r^2)
	\ee
	as $r\to 0$.
\end{prop}

The proof is given in Section~\ref{sec:pfprop2}. Note that the denominators in \eqref{Cdef} and \eqref{phi(r)} are nonzero by \eqref{assumption_c} and \eqref{supercritical}.

\begin{figure}[h]
\begin{center}
	\begin{tikzpicture}[scale=1.4]
	\draw (0,3) node[anchor=north east] {$r$};
	\draw (2,0) node[anchor=north east] {$\varphi$};
	\filldraw[fill=gray!40!white, draw=gray!40!white] (0,0) rectangle (1.5,2.5);
	\draw[thick] (0,2.5) -- (0,0) -- (1.5,0) -- (1.5,2.5);
	\draw[->] (0,0) -- (0,3);
	\draw[->] (0,0) -- (2,0);
	\filldraw (1,0) circle (0.5mm);
	\draw (1,0) to[out=130,in=300] (0.5,0.6) to[out=120,in=250] (0.8,2.5);
	\end{tikzpicture}
	~~~~~
	\begin{tikzpicture}
	\draw[->] (-2,0) -- (2,0);
	\draw[->] (0,-2) -- (0,3);
	\node at (1.8,-0.2) {$x$};
	\node at (-0.2,2.8) {$y$};
	\draw plot[domain=0:460, samples=200] (\x:{0.004*\x});
	\filldraw (0,0) circle (0.3mm);
	\end{tikzpicture}
	~~~~~
	\begin{tikzpicture}[scale=1.4]
	\draw[->] (-1.5,0) -- (1.5,0);
	\draw[->] (0,-0.5) -- (0,2);
	\node at (1.3,-0.2) {$x$};
	\node at (-0.2,1.8) {$z$};
	\draw plot[domain=0:10, samples=200] ({0.1*cos(0.5*\x r)*\x},{0.1*\x});
	\draw[dashed] plot[domain=-1.5:25, samples=200] ({-0.9*exp(-0.2*\x)},{0.9*exp(-0.2*\x)});
	\draw[dashed] plot[domain=-1.5:25, samples=200] ({0.9*exp(-0.2*\x)},{0.9*exp(-0.2*\x)});
	\filldraw (0,0) circle (0.3mm);
	\end{tikzpicture}
\end{center}
\caption{Illustrated is a Bohmian trajectory shortly before/after absorption/emission, asymptotically obeying \eqref{eq:asymp}; the figure is analogous to \cite[Figure~3]{HT22} but shows quite different behavior. LEFT: Drawn in spherical coordinates, with only the azimuthal angle $\varphi$ shown; to leading order near $r=0$, $\varphi(r)=\varphi_0 + Cr$ as in \eqref{phi(r)}; the dot marks $(r=0,\varphi_0)$. 
MIDDLE: The curve $\varphi(r)=\varphi_0 + Cr$ drawn in 2d cartesian coordinates. RIGHT: The curve $\varphi(r)=\varphi_0+Cr$, $\vartheta=\vartheta_0$ drawn in 3d cartesian coordinates, seen along the $y$-axis. Dashed is the cone $\{\vartheta = \vartheta_0\}$.}
\label{fig:infcirc}
\end{figure}
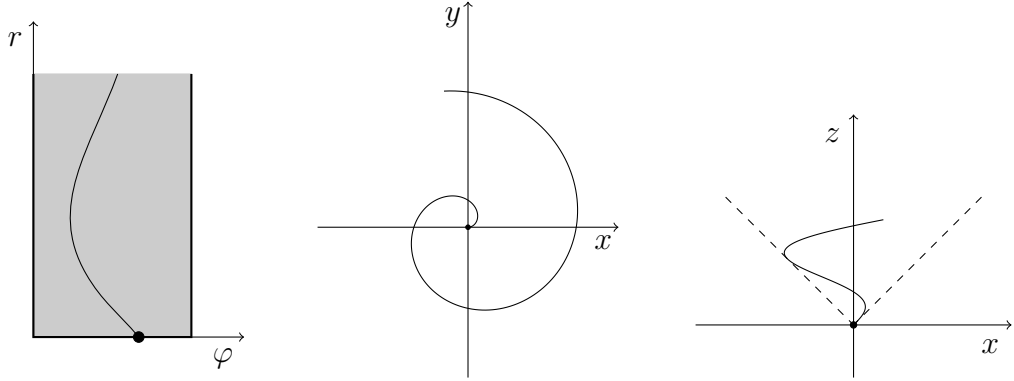

\subsubsection{Bohm-Bell Jump Process} \label{subsec:bohmbell}

We now give the definition of the Bohm-Bell jump process $(Q_t)_{t\geq 0}$ in $\Q$ assuming $\widetilde{\kappa}_j=\pm 1$ and $\Psi_0\in\widehat{D}$. It is a Markov process with the following structure (similar to the ones considered in \cite{timelike,bohmibc,HT22}): 

\paragraph*{Initial Distribution.}
The initial configuration $Q_0$ has probability distribution given by the Born distribution $\PPP_{\Psi_0}$ as in \eqref{Born3}.

\paragraph*{Deterministic Evolution by Bohm's Equation of Motion.} At any time $t$ at which $Q_t$ lies in the upper sector $\Q^{(1)}$, it moves according to Bohm's equation of motion \eqref{eq:Bohmcoo1}; that is, the world line is tangent to the 4-vector field $j^\mu$. 

\paragraph*{Deterministic Jumps.} When $Q_t$ reaches the singularity $r=0$ at time $t_0$, it jumps to the lower sector, $Q_{t_0+}=\emptyset$, and stays there for some time interval.

\paragraph*{Stochastic Jumps.} When $Q_t$ sits in the lower sector, it jumps to a trajectory leaving the singularity with a certain jump rate. The general formula for the rate of jumping at time $t$, given that $Q_{t-}=q'$, to anywhere in an infinitesimal set $\D q$ can be derived \cite{GT05b} to be
\be \label{eq:jumprategeneral}
\sigma_{t}(q' \rightarrow \D q) = \dfrac{\max\{0,J_{\Psi_{t}}^{\perp}(q)\}}{\rho_{\Psi_{t}}(q')}\nu{(\D q,q')}\,,
\ee
where $J^{\perp}$ is the component of probability current in coordinates orthogonal to the boundary of configuration space (in our case, the radial component), $\rho$ is the probability density and $\nu$ the surface area measure on the part of the boundary allowed for jumps from $q'$. In our case, only $q'=\emptyset$ can occur, $\rho_{\Psi_t}(q')=|\Psi^{(0)}_t|^2$, and $\nu(\cdot,\emptyset)$ is the surface area measure on $\mathbb{S}^{2}$ \cite{timelike}. The trajectory onto which to jump gets characterized by the boundary point $q$ at which it starts; in our case, $q$ lies on the boundary $\{0\}\times \SSS^2$ of $[0,\infty)\times \SSS^2$ and thus represents the \emph{direction of emission}. As we will show in Section~\ref{sec:equivariance}, \eqref{eq:jumprategeneral} asserts in our case that the rate of jumping to a point $q$ in the surface element $\{0\}\times \D^2\vomega$ is
\be\label{jumprate}
\sigma_t(\emptyset \to \D^2\vomega) = \dfrac{\max\{0, -\mathrm{Im}[c_{-}^{*}(t)c_{+}(t)] \}}{2\pi|Q| \; \bigl|\Psi_t^{(0)}\bigr|^{2}} \D^2\vomega
\ee
with $c_{\pm}(t)$ from \eqref{cdef}. The total jump rate (or the rate of leaving $\emptyset$) at $t$ is thus given by
\be\label{totaljumprate}
\sigma_t(\emptyset\to \SSS^2) 
=\int_{\vomega\in\SSS^2} \hspace{-3mm} \sigma_t(\emptyset \to \D^2\vomega) = 2\dfrac{\max\{0, -\mathrm{Im}[c_-^*(t)c_+(t)] \}}{|Q|\; \bigl|\Psi_t^{(0)} \bigr|^{2}} \,.
\ee
As we elucidate in Section~\ref{sec:equivariance}, the rate \eqref{jumprate} ensures equivariance. Since the fraction in \eqref{jumprate} does not depend on $\vomega$, the probability distribution of the direction of emission, given that a jump occurs at $t$, is uniform over the sphere.

\bigskip

This completes the definition of the process. We briefly note that the description just given agrees with what was conjectured in \cite{timelike} about the form of the Bohm-Bell process (except that the IBC considered there is not the same as our \eqref{ibc}).
We conclude this section with two remarks.

\begin{remark}[Negative times]
The definition can be extended to provide a process $(Q_t)_{t\in\RRR}$ also for negative times by choosing the initial time to be any $t_0$ instead of $0$, noting that different choices of $t_0$ are compatible with each other (in the sense that the two processes are equal in distribution after the later of the two choices of $t_0$), and letting $t_0\to-\infty$.
\end{remark}

\begin{remark}[Foliation]
	We define the process relative to the foliation given by the Reissner-Nordstr\"om time coordinate, but the random path in space-time is actually indendent of the choice of the foliation. The situation will be different for more than 1 $y$-particle~\cite{HBD}.
\end{remark}

\subsection{Structure of the Following Sections}

The rest of the paper is devoted to proving Theorem \ref{theorem3.1} and justifying our claims on the trajectories and the jump process from Section \ref{subsec:jump}, in particular proving Proposition \ref{prop:2}. In order to do so, we first recall some preliminaries in Section~\ref{chapter2}. Afterwards, in Section~\ref{chapter3} we construct the IBC Hamiltonian and thus prove Theorem~\ref{theorem3.1}.The following Section~\ref{chapter4} deals with the trajectories and the jump process. The ultimate Section \ref{sec:conclusions} contains some concluding remarks.

\section{Preparation of Proofs: Symmetries and Transformations}
\label{chapter2}

In this section we gather some preliminary information regarding the Dirac Hamiltonian in the sRN background. The principal goal of this section is to transform the Hamiltonian in a simple form, thereby exploiting its built-in spherical symmetry (see Section~\ref{lotf} and \cite{CP82,Tha92,BMM00,KTZT20}) and a convenient scalar change of variables (see Section~\ref{subsec:Rcoord}). 
We follow mostly Cohen and Powers \cite{CP82} and Thaller \cite{Tha92}.

\subsection{Radial Symmetry: Hilbert Space Decomposition}
\label{lotf}

We write the Hilbert space $\Hilbert^{(1)}$, given by \eqref{Hilbert1def}, in the form
\be
\Hilbert^{(1)}= L^2\bigl((0,\infty),\CCC,A^{-1}r^2\D r \bigr) \otimes L^2 \bigl(\SSS^2,\CCC^4,\D^2\vomega \bigr)\,.
\ee
As a consequence of its rotational symmetry, $H_1$ leaves angular momentum eigenspaces invariant; in particular, it leaves the specific eigenspaces $L^2((0,\infty),\CCC,A^{-1}r^2\D r) \otimes \Kilbert_{m_j\kappa_j}$ invariant, where
\be
\mathscr{K}_{m_{j},\kappa_{j}} = \bigl\{ c^{+}\Phi_{m_{j},\kappa_{j}}^{+} + c^{-}\Phi_{m_{j},\kappa_{j}}^{-} ~:~ c^{\pm} \in \CCC \bigr\}
\ee
and the $\Phi^{\pm}_{m_j\kappa_j}$ form an ONB of $L^2(\SSS^2,\CCC^4,\D^2\vomega)$ given explicitly in Appendix~\ref{app:Phi}. As a consequence, with respect to the decomposition
\be\label{Angulardecomp}
L^{2}(\SSS^{2},\CCC^{4},\D^2\vomega) = \bigoplus_{j = \frac{1}{2},\frac{3}{2},...}^{\infty} ~ \bigoplus_{m_{j}= -j}^{j} ~ \bigoplus_{\kappa_{j}= \pm(j+\frac{1}{2})} \!\! \Kilbert_{m_{j},\kappa_{j}}\,,
\ee
$H_1$ is block diagonal,
\be
H_1 = \bigoplus_{j = \frac{1}{2},\frac{3}{2},...}^{\infty} ~ 
\bigoplus_{m_{j}= -j}^{j} ~ \bigoplus_{\kappa_{j}= \pm(j+\frac{1}{2})} 
\!\! H_{1m_j\kappa_j}^{\mathrm{red}}\,.
\ee
We consider each block $H_{1m_j\kappa_j}^{\mathrm{red}}$ individually. Relative to the basis $\{\Phi^+_{m_j\kappa_j}, \Phi^-_{m_j\kappa_j}\}$, it can be written as a $2\times 2$ matrix whose entries are operators acting on the radial Hilbert space $L^2((0,\infty),\CCC, A^{-1}r^2\D r)$, in fact
\be\label{Hredr}
H_{1m_j\kappa_j}^{\mathrm{red}} 
= \begin{bmatrix}
        qQr^{-1} + mA & -A^{2}(\partial_{r}+\tfrac{1}{r}) - \dfrac{AA'}{2} + \dfrac{\kappa_j A}{r} \\
        A^{2}(\partial_r+\tfrac{1}{r}) +\dfrac{AA'}{2} + \dfrac{\kappa_j A}{r} &  qQr^{-1} - mA 
    \end{bmatrix}\,.
\ee
(Recall that $A$ is a function of $r$ and $A'$ its derivative.) To see this, we note first that the operators $\alpha^1$ (which in our spinor basis $b_x$ is the $\alpha$ associated with the radial direction) and $\beta$ (and thus also $\gamma^1=\beta\alpha^1$) leave the subspaces $\mathscr{K}_{m_{j}\kappa_{j}}$ invariant; with respect to the basis $\{\Phi^+_{m_j\kappa_j}, \Phi^-_{m_j\kappa_j}\}$, they take the form
\be\label{alphabeta}
\alpha^1 = \begin{pmatrix}
                0 & -\I \\
                \I & 0 
                \end{pmatrix}, \hspace{0.5cm} 
\beta = \begin{pmatrix}
    1 & 0 \\
    0 & -1 \\
\end{pmatrix}, \hspace{0.5cm} 
\gamma^1 = \begin{pmatrix}
                0 & -\I \\
                -\I & 0 
                \end{pmatrix}\,.
\ee
Second, we note that \cite[(2.8) and (2.9)]{CP82}
\be
\gamma^1K = -\alpha^2 (\partial_\vartheta+\tfrac{1}{2} \cot\vartheta) - \alpha^3(\sin\vartheta)^{-1} \partial_\varphi\,.
\ee
With these relations and \eqref{Phieigenvectors}, \eqref{Hredr} follows from \eqref{DiracHamiltonian}.

\subsection{The $R$ Coordinate} \label{subsec:Rcoord}

It turns out useful to change coordinates from $r$ to the ``tortoise coordinate'' (a.k.a. ``Regge-Wheeler coordinate'') $R(r)$, defined for $r\geq 0$ to be the solution to 
\begin{equation}\label{Rdef}
\dfrac{\D R}{\D r} = \dfrac{1}{A(r)^{2}} ~~\text{with}~~R(r=0) = 0
\end{equation}
and $A$ from \eqref{Adef} (see Figure \ref{fig:R(r)}); $R$ is called $x$ in \cite{CP82,Bel98,BMM00} and $r^*$ in \cite{HE73}. It is such that for any fixed $\vomega$, the $(t,R)$ coordinates are ``conformally Lorentzian,'' i.e., the 2d metric in $(t,R)$ coordinates is $1+1$ Minkowskian up to a scalar (conformal) factor $A^2$. As a consequence, any radial null geodesic satisfies $R=t+\mathrm{const.}$, $\vomega=\mathrm{const.}$; thus, one can say the physical meaning of $R$ of a space-time point $x$ is the coordinate time it takes a light ray to reach $x$ from the singularity.
Although we do not need the explicit form of the solution, we mention that it is given by \cite{Bel98,BMM00}
\begin{align}
R(r) 
=r + M \log\left(\dfrac{r^{2}-2Mr +Q^2}{Q^{2}}\right) +\dfrac{2M^{2}-Q^{2}}{\sqrt{Q^{2}-M^{2}}}\arctan\left(\dfrac{r-M}{\sqrt{Q^{2}-M^{2}}}\right) + C \label{Rsolution}
\end{align}
with suitable integration constant $C$.\footnote{The expression given in \cite[p.~157]{HE73} has incorrect constant prefactors.}

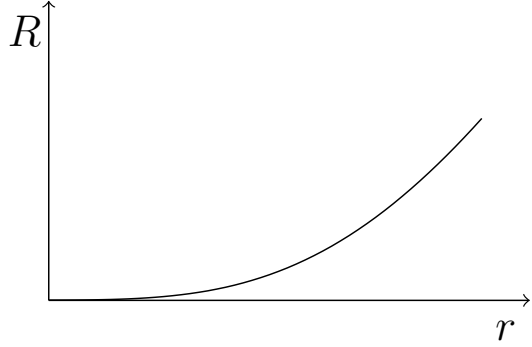
\begin{figure}[h]
\centering
\begin{tikzpicture}[scale=1.4]
	\begin{axis}[
		xtick=\empty,
		ytick=\empty,
		ztick=\empty,
		xmin=-1, xmax=2,
		ymin=-1, ymax= 1,
		samples=1000,
		axis lines = none,
		]
		\draw[->] (0,0) -- (2,0);
		\draw[->] (0,0) -- (0,1);
		\node at (1.9,-0.1) {$r$};
		\node at (-0.1,0.9) {$R$};
		\addplot[domain=0:1.8, samples=50] ({x}, {x + ln(0.25*(x^2 -2*x +4))- 1.15*rad(atan(0.577*(x-1))) -0.601});
		% M= 1, Q=2, C = -0.601
	\end{axis}
\end{tikzpicture}
\caption{Graph of the function $R(r)$ defined in \eqref{Rdef} and given explicitly in \eqref{Rsolution} for $M=1$ and $Q=2$; in this case, $C \approx -0.601$.}
\label{fig:R(r)}
\end{figure}

\begin{lemma}[The $R$-coordinate transformation]\label{Rasymptotics} 
Let $R, r >0$ be related by \eqref{Rdef}. Then 
\begin{subequations}
\begin{align}
&\lim_{R\rightarrow 0} R^{-1/3}\, r(R) = (3Q^{2})^{1/3} \label{RrQ} \\
&\lim_{R\rightarrow 0} r^{2}(R) \, A^{2}(r(R)) = Q^{2}\,.
\end{align}
\end{subequations}
In particular, as $R\to 0$ (or, equivalently, $r\to0$),
\be\label{Rr3}
R \sim r^{3}
\ee
and
\be
A^{2}(r(R)) \sim R^{-2/3}.
\ee
\end{lemma}

\begin{proof}
This follows from \eqref{Rdef} and the definition of $A(r).$ 
\end{proof}

Next, we make use of the $R$ coordinate to define a unitary transformation $U$ of the radial Hilbert space as in \cite{CP82}:
\begin{align}\label{unitarytransform}
	\begin{split}
U: L^2((0,\infty),\CCC, A^{-1}r^2\D r) &\to L^2((0,\infty),\CCC, \D R)\\ 
 \Psi(r) &\mapsto \phi(R) := r(R) \, A(r(R))^{1/2} \, \Psi(r(R)) 
	\end{split}
\end{align}
It is unitary because $\D R = A^{-2} \, \D r$, so $|\phi|^2 \D R= r^2 A |\Psi|^2 A^{-2} \D r = |\Psi|^2 A^{-1} r^2 \D r$.

The main advantage of introducing the $R$ coordinate is that it removes the $A^2$-factor in front of the differential operator in \eqref{Hredr}:

\begin{lemma}[Transformed Hamiltonian] \label{lem:unitarytransform}
Under the unitary transformation $U$ defined by \eqref{unitarytransform}, the reduced Hamiltonian acting on each subspace is given by $h_{m_{j},\kappa_{j}} = UH_{1m_j\kappa_j}^{\mathrm{red}}U^{-1}$ with 
\be \label{eq:hmjkj}
    h_{m_{j},\kappa_{j}} = \begin{bmatrix}
        qQr(R)^{-1} + mA(r(R)) & -\partial_{R} +\kappa_j A(r(R))r(R)^{-1}\\\\
        \partial_{R} + \kappa_j A(r(R))r(R)^{-1} &  qQr(R)^{-1} - mA(r(R))\\
    \end{bmatrix}\,,
\ee
which is well defined and symmetric on the domain
\be \label{eq:hmjkjdomain}
D(h_{m_j, \kappa_j})=C_{c}^{\infty} \bigl( (0,\infty),\mathscr{K}_{m_{j},\kappa_{j}}\bigr)\cong C_{c}^{\infty}\bigl( (0,\infty),\CCC^2 \bigr)\,.
\ee 
\end{lemma}

\begin{proof} 
The formula follows from
\be
\partial_R (U(\Psi)) = U\biggl( \Bigl(A^2\partial_r +\frac{A^2}{r} + \frac{AA'}{2}\Bigr)\Psi \biggr)\,,
\ee
which can be easily verified using $\partial_R = A^2 \partial_r$. Symmetry follows using integration by parts from the fact that $h_{m_j, \kappa_j}$ is the sum of a multiplication operator by a self-adjoint $R$-dependent matrix and $\partial_R$ times the skew-adjoint $R$-independent matrix $\bigl( \begin{smallmatrix} 0&-1\\1&0 \end{smallmatrix} \bigr)$.
\end{proof}

\begin{remark}
From \eqref{eq:r(t)} and \eqref{RrQ}, it follows that
\be\label{R(t)}
R(t) = \frac{C_\mathrm{rad}^3}{3Q^2} \: |t-t_0| + \mathcal{O}\Bigl( |t-t_0|^{4/3} \Bigr) \,.
\ee
\end{remark}

\section{Constructing an IBC: Proof of Theorem \ref{theorem3.1}}
\label{chapter3}

In this section, we construct an IBC Hamiltonian with particle creation and thereby prove Theorem \ref{theorem3.1}.

\subsection{A Family of IBC Hamiltonians: Proof of Theorem \ref{theorem3.1}} \label{subsec:thmref}

Our main result, Theorem \ref{theorem3.1}, will be directly obtained from the following slightly reformulated and generalized version of it. Recall that, the Hilbert space of our model is the mini-Fock space $\Hilbert=\Hilbert^{(0)}\oplus \Hilbert^{(1)}$ with $\Hilbert^{(0)}=\CCC$ and 
$\Hilbert^{(1)}$ given by \eqref{Hilbert1def}.

\begin{theorem}[Generalized reformulation of Theorem \ref{theorem3.1}] \label{thm:thmref}
For every $(\widetilde{m}_j,\widetilde{\kappa}_j)\in\sA$, $g\in\CCC\setminus\{0\}$, and real numbers $a_1, ... , a_4 \in \RRR$ satisfying $a_1a_4 - a_2 a_3 = 1$, there is a self-adjoint operator $H$ with domain $D\subset\Hilbert$ such that 
\begin{enumerate}
\item For every $\Psi\in D$, the upper sector is of the form
\be \label{eq:Psiasympt2}
\Psi^{(1)}(r, \pmb{\omega}) = \frac{c_{-}}{|Q|^{1/2}} \,  r^{-1/2}\Phi_{\widetilde{m}_{j},\widetilde{\kappa}_{j}}^{-}(\pmb{\omega}) +  \sum_{(m_j,\kappa_j ) \in \sA}\frac{c_{+m_{j}\kappa_{j}}}{|Q|^{1/2}} r^{-1/2}\Phi_{m_{j},\kappa_{j}}^{+}(\pmb{\omega}) + \mathcal{O}(r^{1/2})
\ee
as $r\to 0$ for some (uniquely determined) \emph{short-distance coefficients} $c_-, c_{+ m_j \kappa_j} \in \CCC$ and $\Phi^{\pm}$ from \eqref{jointeigenbasis}. 

\item Every $\Psi\in D$ satisfies the IBC \be a_{1}c_{-} + a_{2}c_{+\widetilde{m}_{j}\widetilde{\kappa}_{j}} = g\Psi^{(0)}\label{generalizedIBC} \ee 

\item For every $\Psi^{(1)}\in C_c^\infty((0,\infty)\times \SSS^2, \CCC^4)$, $(0,\Psi^{(1)})\in D$, and $H(0,\Psi^{(1)})=(0,H_1 \Psi^{(1)})$ with $H_1$ as in \eqref{DiracHamiltonian}. Put differently, $(H,D)$ is a self-adjoint extension of $(H_1^0,D^0)$ with $D^0= \{0\}\oplus C_c^\infty((0,\infty)\times \SSS^2, \CCC^4)$ and $H_1^0(0, \psi)=(0,H_1\psi)$.

\item The $0$-particle action of $H$ is given by \be (H\Psi)^{(0)} = g^{*} (a_{3}c_{-}+a_{4}c_{+\widetilde{m}_{j}\widetilde{\kappa}_{j}}) \ee 
	\item Particle creation occurs, i.e., $H$ is not block diagonal in the decomposition $\Hilbert^{(0)}\oplus \Hilbert^{(1)}$.
\end{enumerate}
\end{theorem}

\begin{proof}[Proof of Theorem~\ref{theorem3.1}] 
Theorem \ref{theorem3.1} follows from Theorem \ref{thm:thmref} by taking $a_1, a_4 = 1$, $a_2, a_3 = 0$, and  invoking the particular form of $\Phi^{\pm}$ from \eqref{jointeigenbasis} and $\beta = \mathrm{diag}(1,1,-1,-1)$.
  In particular, the IBC Hamiltonian presented in Theorem \ref{theorem3.1} is in fact part of an entire family of Hamiltonians described by the four real parameters $a_1, ... , a_4$ under the constraint $a_1a_4 - a_2 a_3 = 1$. This concludes the proof of Theorem~\ref{theorem3.1}.
\end{proof}

The rest of this section is devoted to proving Theorem \ref{thm:thmref}.

\begin{remark}[Outline of the proof of Theorem \ref{thm:thmref}]
In constructing the self-adjoint $H$ in Theorem \ref{thm:thmref}, we will decompose the domain $D^{0}$ into fixed angular momentum sectors $\mathscr{K}_{m_{j},\kappa_{j}}$ as in Section~\ref{chapter2}. That is, we will exploit that $D^{0}$ is unitarily equivalent to
\be
\{0\} \oplus \bigoplus\limits_{j,m_{j},\kappa_{j}} C_{c}^{\infty}((0,\infty), \CCC, \D R) \otimes \mathscr{K}_{m_{j},\kappa_{j}}\,. 
\ee
The construction of $H$ now proceeds separately for each sector. We couple one chosen angular momentum sector $\mathscr{K}_{\widetilde{m}_{j},\widetilde{\kappa}_{j}}$ to the $0$-particle sector $\Hilbert^{(0)}$ of $\mathscr{H}$ while all the other angular momentum sectors do not couple to the $0$-particle part. In particular, $H$ is block diagonal relative to the decomposition
\be \label{eq:blockdiag}
\mathscr{H} \cong \widehat{\mathscr{H}} \oplus \bigoplus_{(m_{j},\kappa_{j}) \neq (\widetilde{m}_{j},\widetilde{\kappa}_{j})} L^2((0,\infty), \CCC, \D R) \otimes \mathscr{K}_{m_{j},\kappa_{j}}\,,
\ee
but not relative to
\be \label{eq:hatHilbert}
\widehat{\mathscr{H}} = \mathscr{H}^{(0)} \oplus L^2((0,\infty), \CCC, \D R) \otimes \mathscr{K}_{\widetilde{m}_{j},\widetilde{\kappa}_{j}}
\ee
In the proof, we construct a self-adjoint $\widehat{H}$ acting on $\widehat{\mathscr{H}}$ using interior boundary conditions. This is done by connecting the near-origin behavior of functions in the \emph{adjoint} domain of $C_{c}^{\infty}((0,\infty), \CCC, \D R) \otimes \mathscr{K}_{\widetilde{m}_{j},\widetilde{\kappa}_{j}}$ (see \cite[Theorem 5.2]{CP82}) to the $0$-particle sector $\Hilbert^{(0)}$. In a similar way, we will choose self-adjoint extensions of $h_{m_{j},\kappa_{j}}$ on $C_{c}^{\infty}((0,\infty), \CCC, \D R) \otimes \mathscr{K}_{m_{j},\kappa_{j}}$ for $(m_{j},\kappa_{j}) \neq (\widetilde{m}_{j},\widetilde{\kappa}_{j})$ which do not couple to $\Hilbert^{(0)}$. This completes the construction of a self-adjoint $H$. 
\end{remark}

\subsection{Proof of Theorem \ref{thm:thmref}}
\label{sec:pfthm1}

Throughout the entire proof of Theorem \ref{thm:thmref}, we will heavily use the change of variables from Section \ref{subsec:Rcoord}, i.e., use the coordinate $R$ instead of the usual radial variable $r$, which amounts to the unitary transformation in \eqref{unitarytransform}. Moreover, as a preparation of our proof, we state and prove the following lemma concerning the asymptotic behavior of wave functions $\phi$ in the adjoint domain $D(h_{m_{j}\kappa_{j}}^{*})$ of $D(h_{m_{j}\kappa_{j}})$ from \eqref{eq:hmjkj}--\eqref{eq:hmjkjdomain}

\begin{lemma}\label{Lemma 1} Let $\phi = ( \phi_{+},  \phi_{-}) \in D(h_{m_{j}\kappa_{j}}^{*})$. Then $\phi$ is continuous at $R = 0$, i.e., $\lim\limits_{R\rightarrow 0} \phi_{\pm}(R) = \phi_{\pm}(0)$ exists, and 
\be
\phi_{\pm}(R) = \phi_{\pm}(0) + \mathcal{O}(R^{1/3}) \quad \text{as} \quad R \to 0 \,.
\ee
\end{lemma}

\begin{proof}
In \cite[Lemma 5.1]{CP82}, Cohen and Powers prove that the functions in $D(h_{m_{j}\kappa_{j}}^{*})$ are continuous at $R=0$. Here, we obtain more precise information on their asymptotics. From \cite[Eq.~(5.3)]{CP82}, $\phi_{\pm}(R)$ can be expressed as
\begin{subequations}
\begin{align}
    \phi_{+}(R) &= \E^{\eta(R)}\Bigl(\phi_{+}(0) - \int_{0}^{R}\E^{-\eta(y)} \Bigl( h_{2}(y) - \bigl(m-v_{2}(y)\bigr)\phi_{-}(y)\Bigr)\D y\Bigr) \label{3.12}\\ 
    \phi_{-}(R) &= \E^{-\eta(R)}\Bigl(\phi_{-}(0) + \int_{0}^{R}\E^{\eta(y)}\Bigl(h_{1}(y) - \bigl(m+v_{1}(y)\bigr)\phi_{+}(y) \Bigr)\D y\Bigr)\,,
\end{align}
\end{subequations}
where we denoted $h_{m_{j}\kappa_{j}}^{*}\phi =  (	h_{1}, h_{2})$ and 
\begin{subequations}
\begin{align}
    u(R) &=  \frac{\kappa_j A(r(R))}{r(R)}\\ 
    \eta(R) &= \int_{0}^{R}u(y)\D y  \\[1mm]
    v_{1}(R) &= qQr(R)^{-1} +mA(r(R))- m \\[3mm] 
    v_{2}(R) &= qQr(R)^{-1} -mA(r(R))+ m  
\end{align}  
\end{subequations}
By Lemma \ref{Rasymptotics}, we have, asymptotically as $R \rightarrow 0^{+}$, $v_{i}(R) \sim R^{-1/3}$ and $u(R) \sim R^{-2/3}$. Hence, $\eta(R) = \mathcal{O}(R^{1/3})$ and further $\E^{\eta(R)} = 1 + \mathcal{O}(R^{1/3})$. Now we show that the integral term in \eqref{3.12} contributes only $\mathcal{O}(R^{1/2})$ by estimating the three summands in the integral in \eqref{3.12} separately. 

First, note that since $\phi_{-}(y)$ is a continuous function on the compact interval $[0,R]$, it is bounded. Moreover, since $\eta$ is also bounded, we find that $ \int_{0}^{R} \D y |\E^{-\eta(y)} m  \phi_{-}(y)| = \mathcal{O}(R)$. Next, as $v_{2}(R)= \mathcal{O}(R^{-1/3})$, we obtain $ \int_{0}^{R} \D y |\E^{-\eta(y)}\, v_{2}(y) \, \phi_{-}(y)| = \mathcal{O}(R^{2/3})$.  It thus remains to estimate $\int_{0}^{R} \D y \, \E^{-\eta(y)}\, h_{2}(y)$. 

Instead of the previous $L^\infty$-bounds on the other integrands, we now apply the Cauchy-Schwarz inequality to get 
\be
\left|\int_{0}^{R}\E^{-\eta(y)}h_{2}(y)\D y\right| \leq \|\E^{-\eta}\|_{L^{2}[0,R]}\|h_{2}\|_{L^{2}[0,R]} = \mathcal{O}(R^{1/2})\,,
\ee
where we used that $\|\E^{-\eta}\|_{L^{2}[0,R]} = \mathcal{O}(R^{1/2})$, since $\eta$ is bounded, and $h_2 \in L^2[0,\infty)$, which certainly implies $\Vert h_2 \Vert_{L^2[0,R]} = \mathcal{O}(1)$.\footnote{By the dominated convergence theorem, this can in fact be strengthened to $\Vert h_2 \Vert_{L^2[0,R]} = o(1)$ as $R \to 0$, but we do not follow this improvement for simplicity.}   

Combining all the estimates above, we finally conclude that 
\begin{equation}
\phi_{+}(R) = \phi_{+}(0) + \mathcal{O}(R^{1/3}) \quad \text{as} \quad R \to 0
\end{equation}
as desired.
Similarly, we also get $\phi_{-}(R) = \phi_{-}(0) + \mathcal{O}(R^{1/3})$ as $R \rightarrow 0$. 
\end{proof}

Armed with Lemma \ref{Lemma 1}, we can now turn to the actual proof of Theorem \ref{thm:thmref}. This is divided in three steps: 
\begin{itemize}
\item[(i)] First, in Section \ref{proof1}, we define the domain $D$ of $H$ and show that every $\Psi \in D$ satisfies the asymptotics in \eqref{eq:Psiasympt2} and obeys the IBC \eqref{generalizedIBC}. 
\item[(ii)] In Section \ref{proof2}, we then proceed to show that $H$ acting as in items $3$ and $4$ of Theorem \ref{thm:thmref} is in fact self-adjoint on $D$. 
\item[(iii)] Finally, in Section \ref{proof3}, we prove that particle creation occurs with the so defined Hamiltonian, i.e., it is not block diagonal in the decomposition $\Hilbert^{(0)}\oplus \Hilbert^{(1)}$. 
\end{itemize}

\subsubsection{Definition of the Domain $D$}\label{proof1}

We define the domain $D \subset \Hilbert$ of our Hamiltonian $H$ (to be devised) as
\begin{equation} \label{eq:Ddef}
D:= \big((1 \oplus U^{-1})\widehat{D}_{\widetilde{m}_j \widetilde{\kappa}_j}\big) \oplus \bigoplus_{\substack{j, m_j, \kappa_j \\ (m_j, \kappa_j) \neq (\widetilde{m}_j, \widetilde{\kappa}_j)}} U^{-1} D_{m_j \kappa_j}^{\theta  = 0} \,, 
\end{equation}
where we denoted (recall \eqref{eq:hatHilbert} for the definition of $\widehat{\Hilbert}$)
\begin{equation} \label{eq:hatD}
    \widehat{D}_{\widetilde{m}_j \widetilde{\kappa}_j}:= \Bigl\{ (\Psi^{(0)},\phi^{(1)}) \in \widehat{\mathscr{H}}: \phi^{(1)} \in D(h_{\widetilde{m}_{j},\widetilde{\kappa}_{j}}^{*})\ \text{and} \  a_1 \phi^{(1)}_-(0) + a_2 \phi^{(1)}_+(0) = g \Psi^{(0)} \Bigr\}\,.
\end{equation}
Moreover, for $\theta \in [0,2 \pi)$, we denoted
\begin{equation} \label{eq:Dtheta}
D_{m_j \kappa_j}^{\theta } := \Bigl\{\phi = (\phi_+, \phi_-) \in D(h_{m_{j},\kappa_{j}}^{*}) ~:~ \phi_{+}(0)\sin\theta +\phi_{-}(0)\cos\theta = 0\Bigr\}\,.
\end{equation}

This means that, for $\theta = 0$ and $(m_j, \kappa_j) \neq (\widetilde{m}_j, \widetilde{\kappa}_j)$, $c_{- m_j \kappa_j} := \phi_-(0)  = 0$ and $c_{+ m_j \kappa_j} := \phi_+(0) \in \CCC$ is free. We also denote $c_{\pm \widetilde{m}_j, \widetilde{\kappa}_j}  := \phi^{(1)}_\pm(0)$ for $(\widetilde{m}_j, \widetilde{\kappa}_j)$ as in \eqref{eq:hatD} and  abbreviate $c_- \equiv c_{- \widetilde{m}_j, \widetilde{\kappa}_j} $. Therefore, inverting the unitary transform $U$ from \eqref{unitarytransform} in \eqref{eq:Ddef} and invoking Lemma~\ref{Rasymptotics}, we find that for every $\Psi \in D$, the upper sector part $\Psi^{(1)}$ obeys the asymptotics given in \eqref{eq:Psiasympt2}. Moreover, inverting $U$ again, we also find that, by definition of  $\widehat{D}_{\widetilde{m}_j \widetilde{\kappa}_j}$, every $\Psi \in D$ obeys the IBC from \eqref{generalizedIBC}. This proves items $1$ and $2$ of Theorem~\ref{thm:thmref}.

\subsubsection{Self-adjointness of $H$ on $D$}\label{proof2}

First, we have that $h_{m_{j},\kappa_{j}}$ on $D(h_{m_j \kappa_j})$ from \eqref{eq:hmjkjdomain} has self-adjoint extensions parametrized by $\theta \in [0,2\pi)$ as \cite[Theorem~5.2]{CP82}   
\be\label{Theorem 5.2}
h_{m_{j},\kappa_{j}}^{\theta} = h_{m_{j},\kappa_{j}}^{*}\Big|_{D^\theta_{m_j \kappa_j}}\,,
\ee
where $D^\theta_{m_j \kappa_j}$ is defined in \eqref{eq:Dtheta}. 
Therefore, since $H$ leaves the decomposition into angular momentum subspaces invariant, the task of proving self-adjointness of $H$ on $D$ immediately simplifies: It reduces to proving that the Hamiltonian $\widehat{H} \equiv \widehat{H}_{\widetilde{m}_j \widetilde{\kappa}_j}$ acting on $\phi = (\phi^{(0)}, \phi^{(1)}) \in \widehat{D} \equiv \widehat{D}_{\widetilde{m}_j \widetilde{\kappa}_j}$ from \eqref{eq:hatD} with $\phi^{(0)} \equiv \Psi^{(0)}$ as (recall the notation below~\eqref{eq:Dtheta})
\begin{subequations}
	\begin{align}
		(\widehat{H}\phi)^{(0)} &= g^{*} [a_{3}c_{-}+ a_{4}c_{+}] \\
		(\widehat{H}\phi )^{(1)} &= h^{*} \phi^{(1)} \,,
	\end{align}
\end{subequations}
is self-adjoint. Here and in the following, to ease notation, we denote $h \equiv h_{\widetilde{m}_{j},\widetilde{\kappa}_{j}}$ as well as $\Kilbert = \Kilbert_{\widetilde{m}_{j},\widetilde{\kappa}_{j}}$

The proof of $(\widehat{H}, \widehat{D})$ being self-adjoint is very similar to \cite[p.~12--13]{HT20}, hence we will be quite brief. First, the fact that $\widehat{D} \subset \widehat{\Hilbert}$ is dense, can be seen in exactly the same way as in \cite{HT20}. 

Next, in order to show that $\widehat{H}$ is symmetric on $\widehat{D}$, we take, completely analogously to \cite[Eqs.~(73)--(86)]{HT20}, some $\phi, \eta \in \widehat{D}$ and compute the difference $\langle \phi, \widehat{H}\eta\>\rangle_{\widehat{\Hilbert}} - \langle \widehat{H}\phi, \eta\>\rangle_{\widehat{\Hilbert}}$. Denoting $c_\pm = \eta_\pm(0)$ and $d_\pm = \phi_\pm (0)$, and using that $a_1 a_4 - a_2 a_3 = 1$, we find
\begin{equation}
\begin{split}
\langle \phi, \widehat{H}\eta\>\rangle_{\widehat{\Hilbert}} &- \langle \widehat{H}\phi, \eta\>\rangle_{\widehat{\Hilbert}} \\
&=  \langle \phi^{(1)}, h^{*}\eta^{(1)}\rangle_{L^2 ((0,\infty),\mathscr{K})} 
-  \langle h^{*}\phi^{(1)}, \eta^{(1)}\rangle_{L^2 ((0,\infty),\mathscr{K})}- [d^{*}_{+}c_{-} - d^{*}_{-}c_{+}]\,, 
\end{split}
\end{equation}
just as in \cite{HT20}. To se that $\langle \phi, \widehat{H}\eta\>\rangle_{\widehat{\Hilbert}} = \langle \widehat{H}\phi, \eta\>\rangle_{\widehat{\Hilbert}}$, we are now left to compute
\begin{align*}
  &\langle \phi^{(1)}, h^*\eta^{(1)}\rangle_{L^2 ((0,\infty),\mathscr{K})} -\langle h^*\phi^{(1)}, \eta^{(1)}\rangle_{L^2 ((0,\infty),\mathscr{K})} \\
  &=\int_{0}^{\infty}\D R\hspace{0.2cm} \partial_{R} \Bigl[ \phi_{-}^{(1)}(R)^{\dag} \: \eta_{+}^{(1)}(R)-\phi_{+}^{(1)}(R)^{\dag} \: \eta_{-}^{(1)}(R) \Bigr]\\[2mm]
  &=\left[\phi_{-}^{(1)}(R)^{\dag} \: \eta_{+}^{(1)}(R)-\phi_{+}^{(1)}(R)^{\dag} \: \eta_{-}^{(1)}(R)\right]_{0}^{\infty}\\[2mm]
  &= \lim\limits_{R\searrow 0}\left[\phi_{+}^{(1)}(R)^{\dag} \: \eta_{-}^{(1)}(R)-\phi_{-}^{(1)}(R)^{\dag} \: \eta_{+}^{(1)}(R)\right]\\[2mm]
  &= d^{*}_{+}c_{-} - d^{*}_{-}c_{+}\,,
\end{align*}
where in the first step we employed that all the terms not involving the derivative $\partial_R$ cancel (cf.~\cite[Eqs.~(75)--(79)]{HT20}). In the penultimate step, we used that $\phi^{(1)}_\pm, \eta^{(1)}_\pm$ vanish at infinity (as follows from them being continuous and in $L^2$). Finally, in the last step we used the IBC in the form of \eqref{eq:hatD}. 

After having proven that $\widehat{H}$ is symmetric on $\widehat{D}$, it remains to show that $\widehat{D}=D(\widehat{H}^{*})$.
In order to do so, first note that $\widehat{D} \subseteq D(\widehat{H}^{*}) \subseteq \CCC \oplus D(h^{*})$. 
Any given $\phi\in\CCC\oplus D(h^{*})$ lies in $D(\widehat{H}^{*})$ if and only if there exists some $\xi \in\widehat{\Hilbert}$ such that for every $\eta\in\widehat{D}$, it holds that $ \langle\xi, \eta \rangle_{\widehat{\Hilbert}} = \langle \phi, \widehat{H} \eta \rangle_{\widehat{\Hilbert}}$. The right-hand side can now be computed, completely analogously to \cite[Eqs.~(89)--(94)]{HT20}, as 
\begin{equation} \label{eq:sa}
\begin{split}
\langle \phi, \widehat{H} \eta \rangle_{\widehat{\Hilbert}} 
= \bigl[ -(d_{-}a_{1} + &d_{+}a_{2}) + g\phi^{(0)} \bigr]^{*}
(a_{3}c_{-} +a_{4}c_{+}) \\
& + \langle h^{*}\phi^{(1)}, 
\eta^{(1)}\rangle_{L^2 ((0,\infty),\mathscr{K})} 
+ \langle g^{*}(a_{3}d_{-} + a_{4}d_{+}), \eta^{(0)}\rangle_{\CCC} 
\end{split}
\end{equation}
where we again abbreviated $c_\pm = \eta_\pm(0)$ and $d_\pm = \phi_\pm (0)$. 
From \eqref{eq:sa} we conclude that $ \langle\xi, \eta \rangle_{\widehat{\Hilbert}} = \langle \phi, \widehat{H} \eta \rangle_{\widehat{\Hilbert}}$ is true for all $\eta \in \widehat{D}$, if and only if 
\begin{equation}
    \xi^{(0)} = g^{*}(a_{3}d_{-} + a_{4}d_{+}) \quad \text{and} \quad 
   \xi^{(1)} = h^{*}\phi^{(1)},
\end{equation}
and $\phi$ satisfies the IBC
\begin{equation}
a_{1}d_{-} + a_{2}d_{+} = g\phi^{(0)}\,.
\end{equation}
This means, $\phi \in \widehat{D}$ and $\xi = \widehat{H} \phi$, i.e., $\widehat{H}$ is self-adjoint on $\widehat{D}$.

\subsubsection{Particle Creation} \label{proof3}

Assume that particle creation did \emph{not} occur, i.e., that the Hamiltonian were block diagonal in the decomposition  $\Hilbert^{(0)} \oplus \Hilbert^{(1)}$, say
\begin{equation}
H = \left(
\begin{array}{c|c}
	F_{0} & 0 \\ \hline
	0 & F_{1}
\end{array}
\right)\,,
\end{equation}
where $F_{0}$ and $F_{1}$ are blocks that act on $\Hilbert^{0}$ and $\Hilbert^{(1)}$ respectively. Under this assumption, the domain of $H$ would be the Cartesian product of the domain of $F_{0}$ (which must be $\Hilbert^{(0)}$) and the domain of $F_{1}$ (a dense subspace of $\Hilbert^{(1)}$). 
Thus, for any $\Psi^{(0)} \in \CCC \setminus \{0\}$, $(\Psi^{(0)}, \Psi^{(1)} \equiv 0)$ is in the domain of a block-diagonal $H$. 
On the other hand, wave functions in the domain of $H$ must satisfy the IBC \eqref{generalizedIBC}, which implies that, since $\Psi^{(0)} \neq 0$, $\Psi^{(1)}$ cannot be identically equal to zero. This is a contradiction, and hence the IBC forces $H$ to be non block-diagonal and we have thus proven item $5$ of Theorem \ref{thm:thmref}. 

This concludes the proof of Theorem \ref{thm:thmref}. \qed

\section{Creation Rate and Trajectories: Proof of Proposition~\ref{prop:2}}
\label{chapter4}

In this section, we verify the claims from Section \ref{subsec:jump}. To this end, we compute the asymptotics of the probability current $j^\mu$ in Proposition \ref{prop:1} in Section \ref{subsec:jproof}. Afterwards, in Section \ref{sec:pfprop2}, we give the proof of Proposition \ref{prop:2}, yielding the asymptotic behavior of the trajectories as solutions to the simplified Bohmian equation of motion \eqref{eq:BohmEMapp}. Finally, in Section \ref{sec:equivariance}, we (non-rigorously) verify that the Bohm-Bell jump process defined in Section \ref{subsec:bohmbell} is equivariant. 

As in Section~\ref{subsubsec:trajectories}, we will consider only the Hamiltonian provided by Theorem~\ref{theorem3.1} (i.e., $a_1=1= a_4$, $a_2=0= a_3$ in the notation of Theorem~\ref{thm:thmref}) and only wave functions $\Psi$ from $\widehat{D} \subset \Hilbert$ as in \eqref{hatDdef}, an invariant subspace comprising $\Hilbert^{(0)}$ and $\Kilbert_{\widetilde{m}_j\widetilde{\kappa}_j}$.

\subsection{Probability Current} \label{subsec:jproof}

In the following proposition, we provide the asymptotic behavior of the probability current $j^\mu$. Recall that
\be
j^0=|\Psi^{(1)}|^2\,, \quad j^i=\Psi^{(1)\dagger} \alpha^i \Psi^{(1)} ~~\text{for}~i=1,2,3.
\ee

\begin{prop}[Asymptotic behavior of the current] \label{prop:1}
	Let $\Psi \in \widehat{D}$ and let $c_\pm$ be defined as in \eqref{cdef}. Then the components of the probability current $j^\mu$ defined in \eqref{jdef} in the basis $e_x$ of \eqref{edef} obey the asymptotics (as $r\to 0$)
\begin{subequations} \label{eq:jproof}
\begin{align}
	 j^0(r,\vartheta,\varphi) &= \frac{|c_+|^2 + |c_-|^2}{4\pi|Q|} r^{-1} + \mathcal{O}(r^0) \label{j0}\\
     j^1(r,\vartheta,\varphi) &= -\dfrac{\mathrm{Im}[c_{-}^{*}c_{+}]}{2\pi |Q|} r^{-1} +\mathcal{O}(r^0) \label{j1}\\
     j^2(r,\vartheta,\varphi) &= \mathcal{O}(r^0) \label{j2}\\
     j^3(r,\vartheta,\varphi) &= \mathrm{sgn}(\widetilde{m}_{j}\widetilde{\kappa}_{j})\sin\vartheta \dfrac{\mathrm{Re}[c_{-}^{*}c_{+}]}{2\pi|Q|} r^{-1} + \sin \vartheta \, \mathcal{O}(r^0)\,. \label{j3}
\end{align}
\end{subequations}
\end{prop}

\begin{proof}
By \eqref{eq:Psiasympt2} and $\Psi\in\widehat{D}$,
\be
\Psi^{(1)}(r,\vomega) 
= \biggl( \frac{c_-}{|Q|^{1/2}} \Phi^-_{\widetilde{m}_j\widetilde{\kappa}_j}(\vomega)  
+ \frac{c_+}{|Q|^{1/2}} \Phi^+_{\widetilde{m}_j\widetilde{\kappa}_j}(\vomega) \biggr) r^{-1/2} 
+ \mathcal{O}(r^{1/2})\,.
\ee
Eq.~\eqref{j0} follows from the facts that 
\be
\langle \Phi^+_{m_j\kappa_j}(\vomega),\Phi^-_{m_j\kappa_j} (\vomega)\rangle_{\CCC^4}=0~~\forall\vomega\in\SSS^2\,,
\ee
that $\{w_1,w_2\}$ is orthonormal in $\CCC^2$, and that 
\be
\bigl|\Phi^{\pm}_{\widetilde{m}_j\widetilde{\kappa}_j}(\vomega)\bigr|^2 = \frac{1}{4\pi}~~~\forall\vomega\in\SSS^2
\ee
for $\widetilde{\kappa}_j=\pm 1$ (so $j=\tfrac{1}{2}$), which can be easily verified from the definition \eqref{eq:Psimjj} using 
\be
Y^0_0(\vartheta,\varphi)=\frac{1}{\sqrt{4\pi}},~~
Y^{\pm 1}_1(\vartheta,\varphi)=\pm \sqrt{\frac{3}{8\pi}} \E^{\pm\I \varphi} \sin \vartheta,~~
Y^0_1(\vartheta,\varphi)= \sqrt{\frac{3}{4\pi}}\cos\vartheta
\,.
\ee 

We turn to \eqref{j1}--\eqref{j3}. Recalling that
\be
\pmb{\alpha} = \begin{pmatrix}
    0  &  \pmb{\sigma} \\
    \pmb{\sigma}  &  0
\end{pmatrix}\,,
\ee
one sees from \eqref{jointeigenbasis} that
\be
\langle \Phi^{\pm}_{{m}_j{\kappa}_j}(\vomega), \alpha^i \Phi^{\pm}_{{m}_j{\kappa}_j} (\vomega)\rangle_{\CCC^4} = 0~~~\forall \vomega\in\SSS^2 \:\forall i=1,2,3.
\ee
Further calculations show that
\begin{subequations}
\begin{align}
    & \langle\Phi_{\widetilde{m}_{j},\widetilde{\kappa}_{j}}^{+}(\pmb{\omega}), \alpha^1\Phi_{\widetilde{m}_{j},\widetilde{\kappa}_{j}}^{-}(\pmb{\omega})\rangle_{\CCC^{4}} = \dfrac{-\I}{4\pi} \,, \label{proofprop1} \\
    &\langle\Phi_{\widetilde{m}_{j},\widetilde{\kappa}_{j}}^{+}(\pmb{\omega}), \alpha^2\Phi_{\widetilde{m}_{j},\widetilde{\kappa}_{j}}^{-}(\pmb{\omega})\rangle_{\CCC^{4}} = 0 \,, \label{proofprop2}\\
    &\langle\Phi_{\widetilde{m}_{j},\widetilde{\kappa}_{j}}^{+}(\pmb{\omega}), \alpha^3\Phi_{\widetilde{m}_{j},\widetilde{\kappa}_{j}}^{-}(\pmb{\omega})\rangle_{\CCC^{4}} = \mathrm{sgn}(\widetilde{m}_{j}\widetilde{\kappa}_{j}) \, \sin\vartheta\dfrac{1}{4\pi}\label{proofprop3} \,. 
\end{align}
\end{subequations}
(In fact, this follows from \cite[Eq.~(47)]{HT22} and \eqref{Wsigma}.) From these relations, \eqref{j1}--\eqref{j3} follow.
\end{proof}

\subsection{Bohmian Trajectories: Proof of Proposition \ref{prop:2}}
\label{sec:pfprop2}

From \eqref{eq:jproof} at $t_0$ while assuming \eqref{assumption_c} (in particular $c_-\neq 0 \neq c_+$) together with \eqref{Bohmcoo2}, we obtain for the approximate trajectories (i.e., the solutions to \eqref{eq:BohmEMapp}), analogously to \cite[Eq.~(60)]{HT22}, that
\begin{subequations}\label{Bohmcoo3}
\begin{align}
\frac{\D r(t)}{\D t}&= -\frac{2Q^2\,\Im[c_-^*c_+]}{|c_+|^2+|c_-|^2} \:  r^{-2}+\mathcal{O}(r^{-1}) \label{radialvel}\\
\frac{\D \vartheta(t)}{\D t}&= \mathcal{O}(r^{-1}) \label{polarvel}\\
\frac{\D \varphi(t)}{\D t}&=  \mathrm{sgn}(\widetilde{m}_{j}\widetilde{\kappa}_{j}) \dfrac{2Q\,\Re[c_-^*c_+]}{|c_+|^2+|c_-|^2} \: r^{-2} +\mathcal{O}(r^{-1})\,. \label{azimuthalvel}
\end{align}
\end{subequations}
Similarly to the arguments given at the end of Section~4 in \cite{HT22}, the differential equations \eqref{Bohmcoo3} can be solved by a simple separation of variables, where one first solves \eqref{radialvel} and then feeds the result into the other two relations, eventually yielding \eqref{eq:asymp}. \qed

\subsection{Equivariance of the Bohm-Bell Process}
\label{sec:equivariance}

In this section, we non-rigorously verify that the process $Q_t$ is equivariant. 

First, away from the origin, we must have equivariance by conservation of probability expressed through the continuity equation
\be
\nabla_\mu j^\mu=0\,.
\ee

Therefore, the only place where probability is gained or lost is at the singularity $r=0$. Consider the probability flux through the surface element $\D^{2}\pmb{\omega}$ near $r =0$ in coordinate space $[0,\infty)\times \SSS^2$, which is $v^1(r,\vomega) \,\rho(r,\vomega)\, \D^2\vomega$ with probability density $\rho$ in coordinate space given by $|\Psi^{(1)}(r,\vomega)|^2 A^{-1} r^2$ according to \eqref{Born3b}. By \eqref{Bohmcoo2}, $v^1=A^2 j^1/j^0$. Thus, the flux is
\be
Ar^2 \Psi^{(1)}(r,\vomega)^\dagger \alpha^1 \, \Psi^{(1)}(r,\vomega) \, \D^2\vomega 
\ee
which converges, as $r\to 0$, to
\be\label{flux}
J_\mathrm{rad} \, \D^2\vomega:= -\frac{\Im[c_-^*c_+]}{2\pi|Q|} \, \D^2\vomega
\ee
by \eqref{j1}. This is the quantity $J_{\Psi_{t}}^{\perp}(q) \, \nu(\D q, q')$ of \eqref{eq:jumprategeneral}. Thus, the rate of gain (positive or negative) of probability at the singularity is given by $4\pi J_{\mathrm{rad}}$.

This agrees with the rate of gain (positive or negative) of probability at $r=0$ of $Q_{t}$: Indeed, in case that $J_{\mathrm{rad}} > 0$, then no trajectory ends at the origin (so no probability is lost) and the amount transported by jumps from $\emptyset$ to the trajectories emanating at time $t_0$ is given by the probability at $\emptyset$ times the total jump rate \eqref{totaljumprate} from $\emptyset$, i.e.,
\begin{equation}
\bigl|\Psi_{t_{0}}^{(0)}\bigr|^{2} \sigma_{t_0}(\emptyset \to \SSS^2) 
= \frac{2}{|Q|}\max\{0, -\mathrm{Im}[c_{-}^{*}(t_{0})c_{+}(t_{0})] \}= 4 \pi J_{\mathrm{rad}} \,.
\end{equation}
In the contrary case, $J_{\mathrm{rad}} < 0$, then no upward jump occurs (and thus no probability is gained at the origin) and the lost amount of probability automatically agrees with the flux across the sphere (since $Q_t$ is $|\Psi_t|^2$-distributed). 
 
Finally, in order to ensure preservation of the $|\Psi|^2$-distribution, it remains to check that the distribution of $Q_{t}$ over the emanating trajectories yields the flux \eqref{flux} through $\D^{2}\pmb{\omega}$ in the $r \to 0$ limit. This follows from the fact that both the flux \eqref{flux} and the jump rate \eqref{jumprate} are uniform over the sphere. 
This concludes our argument for equivariance.

\section{Conclusions}
\label{sec:conclusions}

In this work, we have considered a model of particle creation and annihilation at the singularity of the sRN space-time that avoids the problem of ultraviolet divergence by using interior-boundary conditions. Furthermore, we constructed the corresponding Bohm-Bell process, an equivariant Markov jump process defined through 2 equations: Bohm's equation of motion \eqref{eq:BohmEM} and the formula \eqref{jumprate} which dictates the rate at which particle creation occurs.

For further research, one can consider full Fock space, including particle sectors with more than $1$ particle. It would also be of interest to prove the existence of the Bohm-Bell process, and to define it also for $\widetilde{\kappa}_j\neq \pm 1$ and/or wave functions outside the subspace $\widehat{D}$. We expect the process to be qualitatively similar in these other cases. Furthermore, it would be interesting to consider the case of space-time singularities other than that of sRN.

\appendix

\section{The $\Phi_{m_j\kappa_j}$ in Spherical Coordinates}
\label{app:Phi}

In Minkowski space-time, let $\tilde e$ be an orthonormal basis (Lorentz frame) and $\tilde b$ the corresponding basis in 4d Dirac spin space $S$. Now for spherical coordinates $r\in(0,\infty),\vartheta\in[0,\pi],\varphi\in[0,2\pi)$, let
\begin{subequations}
\begin{align}
\ve_r &=(\sin\vartheta \cos\varphi, \sin\vartheta \sin\varphi, \cos \vartheta)\\
\ve_\vartheta &= (\cos\vartheta \cos\varphi, \cos\vartheta \sin\varphi, -\sin \vartheta)\\
\ve_\varphi &= (-\sin\varphi, \cos \varphi,0)
\end{align}
\end{subequations}
be the orthonormal basis of $\RRR^3$ whose vectors point in the directions of increasing $r,\vartheta,\varphi$ coordinates. Together with the timelike vector of $\tilde e$, they form another, $(r,\vartheta,\varphi)$-dependent Lorentz frame $e$; let $b$ be the corresponding basis of $S$. Then, for any element of $S$, its (spherical) $b$-coefficients are obtained from the (Cartesian) $\tilde b$-coefficients through multiplication by
\be
\begin{bmatrix} W & 0 \\ 0 & W \end{bmatrix}
\ee
with the unitary $2\times 2$ matrix
\be\label{Wdef}
W:= \frac{1}{\sqrt{2}} 
\begin{bmatrix} 
\I\E^{\I(\vartheta+\varphi)/2} & \quad\; \E^{\I(\vartheta-\varphi)/2}\\
\I\E^{\I(-\vartheta+\varphi)/2} & -\E^{\I(-\vartheta-\varphi)/2} 
\end{bmatrix}  \,,
\ee
(whose columns will be denoted by $w_1$ and $w_2$). This follows from the easily verifiable facts that, for $\vsigma=(\sigma_1,\sigma_2,\sigma_3)$ the triple of Pauli matrices,
\begin{subequations}\label{Wsigma}
\begin{align}
W^{-1}\sigma_1W &= \ve_r\cdot \vsigma\\
W^{-1}\sigma_2W &= \ve_\vartheta\cdot \vsigma\\
W^{-1}\sigma_3W &= \ve_\varphi\cdot \vsigma\,,
\end{align}
\end{subequations}
which shows that 2-spinors transform according to $W$, together with the fact that spatial rotations are implemented on 4-spinors as block diagonal $4\times 4$ matrices with $2\times 2$ blocks that are equal to each other and given by the action of the rotation on 2-spinors \cite[(2.172) and (1.38)]{Tha92}.

Relative to the Cartesian basis $\tilde{b}$ in $S$, the explicit form of the functions $\Phi^{\pm}_{m_j\kappa_j}$ is given in \cite[Sec.~4.6.4]{Tha92}; translated into the spherical basis $b$, they are given as follows:
\be\label{jointeigenbasis}
\Phi_{m_{j},\mp(j+\frac{1}{2})}^{+} = \begin{pmatrix}
    \mathrm{i}\Psi_{j\mp\frac{1}{2}}^{m_{j}}\\
    0
\end{pmatrix}, \hspace{1cm} \Phi_{m_{j},\mp(j+\frac{1}{2})}^{-} = \begin{pmatrix}
    0 \\
    \I\Psi_{j\pm\frac{1}{2}}^{m_{j}} 
\end{pmatrix} \,,
\ee
where
\begin{subequations}\label{eq:Psimjj}
\begin{align} 
    \Psi_{j-\frac{1}{2}}^{m_{j}}
    &= \sqrt{\dfrac{j+m_j}{2j}}\: Y_{j- 1/2}^{m_{j}-1/2} w_1
        +\sqrt{\dfrac{j-m_j}{2j}} \: Y_{j- 1/2}^{m_{j}+1/2} w_2
    \\
    \Psi_{j+\frac{1}{2}}^{m_{j}} 
    &= \sqrt{\dfrac{j+1-m_j}{2j+2}} \: Y_{j+1/2}^{m_{j}-1/2} w_1
        -\sqrt{\dfrac{j+1+m_j}{2j+2}} \: Y_{j+ 1/2}^{m_{j}+1/2} w_2
\end{align}
\end{subequations}
with $Y_{\ell}^{m}$ the usual spherical harmonics (e.g., \cite[Sec.~4.6.4]{Tha92}), defined for $\ell \in \mathbb{N}\cup\{0\}$ and $m \in \{-\ell,\ldots,\ell\}$ (not to be confused with the mass in the Dirac equation).

\bigskip

\noindent{\it Acknowledgments.} We thank Michael Kiessling, A. Shadi Tahvildar-Zadeh, and Stefan Teufel for helpful discussions.

\section*{Declarations}

\noindent{\it Funding.} JH gratefully acknowledges partial financial support by the ERC Advanced Grant ``RMTBeyond'' No. 101020331.

\noindent{\it Conflict of interests.} The authors declare no conflict of interest.

\noindent{\it Availability of data and material.} Not applicable.

\noindent{\it Code availability.} Not applicable.

\end{document}